%% file: graph_states_new_hope_arxiv.tex
\documentclass[a4paper,11pt]{article}

\usepackage{fullpage,amsthm,color}
\usepackage{titlesec,authblk}

\definecolor{myurlcolor}{rgb}{0,0,0.4}
\definecolor{mycitecolor}{rgb}{0,0.7,0}
\definecolor{myrefcolor}{rgb}{0.7,0,0}
\usepackage{hyperref}
\hypersetup{colorlinks,
linkcolor=myrefcolor,
citecolor=mycitecolor,
urlcolor=myurlcolor}

\usepackage{amssymb, amsmath, amsthm, verbatim}
\usepackage[mathscr]{euscript}
\usepackage[draft]{fixme}
\usepackage{tikz}
\usepackage{subfigure}
\usepackage{multicol}
\usepackage{amsmath,bbm}
\usepackage{graphicx}
\usepackage{amsfonts}
\usepackage{amssymb}
\usepackage{amsmath, amssymb, amsthm,verbatim,graphicx,bbm}
\usepackage{color,xcolor,longtable}
\usepackage{palatino}
\usepackage{mathpazo}
\usepackage{bbm, dsfont}

\usepackage{hyperref}

\usetikzlibrary{patterns}
\newcommand{\beq}[0]{\begin{equation}}
\newcommand{\eeq}[0]{\end{equation}}

\newcommand{\I}{\mathbbm{1}}
\newcommand{\w}{\omega}
\newcommand{\ket}[1]{|#1\rangle}
\newcommand{\bra}[1]{\langle#1|}

\newcommand{\B}{\mathcal{B}}

\newcommand{\va}{\Vec{a}}
\newcommand{\vx}{\Vec{x}}

\newcommand{\one}{\leavevmode\hbox{\small1\normalsize\kern-.33em1}}

\def\be{\begin{equation}}
\def\ee{\end{equation}}
\def\ben{\begin{eqnarray}}
\def\een{\end{eqnarray}}
\def\eea{\end{array}}
\def\bea{\begin{array}}

\newcommand{\Tr}[1]{\mathrm{Tr}#1}
\newcommand{\bei}{\begin{itemize}}
\newcommand{\eei}{\end{itemize}}

\newcommand{\proj}[1]{\ket{#1}\!\bra{#1}}

\def\B{{\cal B}}

\newcommand{\N}{\mathbb{N}}

\renewcommand{\emph}[1]{\textbf{#1}}

\theoremstyle{plain}
\newtheorem{thm}{Theorem}

\newtheorem{fakt}{Fact}

\newtheorem{cor}[thm]{Corollary}

\newtheorem{defn}[thm]{Definition}

\theoremstyle{definition}

\theoremstyle{remark}

\titleformat*{\section}{\large\bfseries}

\begin{document}

\title{Scalable Bell inequalities for graph states of arbitrary prime local dimension and self-testing}
\author{Rafael Santos$^{1}$, Debashis Saha$^2$, Flavio Baccari$^{3}$, Remigiusz Augusiak$^1$}

\affil{
$^1$Center for Theoretical Physics, Polish Academy of Sciences, Aleja Lotnik\'ow 32/46, 02-668 Warsaw, Poland\\
$^2$ School of Physics, Indian Institute of Science Education and Research Thiruvananthapuram, Kerala, India 695551\\
$^3$Max-Planck-Institut f\"ur Quantenoptik, Hans-Kopfermann-Stra{\ss}e 1, 85748 Garching, Germany
}
\renewcommand\Affilfont{\itshape\small}


\date{}

\maketitle

\begin{abstract}
Bell nonlocality---the existence of quantum correlations that cannot be explained by classical means---is certainly one of the most striking features of quantum mechanics. Its range of applications in device-independent protocols is constantly growing. Many relevant quantum features can be inferred from violations of Bell inequalities, including entanglement detection and quantification, and state certification applicable to systems of arbitrary number of particles. A complete characterisation of nonlocal correlations for many-body systems is, however, a computationally intractable problem. Even if one restricts the analysis to specific classes of states, no general method to tailor Bell inequalities to be violated by a given state is known. In this work we provide a general construction of Bell inequalities that are maximally violated by graph states of any prime local dimension. These form a broad class of multipartite quantum states that have many applications in quantum information, including quantum error correction. We analytically determine their maximal quantum violation, a number of high relevance for device-independent applications of Bell inequalities. Finally, we show that these inequalities can be used for self-testing of multi-qutrit graph states such as the well-known four-qutrit absolutely maximally entangled state AME(4,3).
\end{abstract}

\section{Introduction}

The first Bell inequalities were introduced to show that certain predictions of quantum theory cannot be explained by classical means \cite{bell1964einstein}. In particular, correlations obtained by performing local measurements on joint entangled quantum states are able to violate Bell inequalities and hence cannot arise from a local hidden variable model (LHV). The existence of such non-local correlations is referred to as Bell non-locality or simply non-locality. 

Since then the range of applications of Bell inequalities has become much wider. In particular, they can be used
for certification of certain relevant quantum properties
in a device-independent way, that is, under minimal assumptions about the underlying quantum system. First,
violation of Bell inequalities can be used to certify the
dimension of a quantum system \cite{DimWit} or the amount of entanglement present in it \cite{Moroder2013}. Then, Bell violations are used to certify that the outcomes of quantum measurements are truly random \cite{Pironio2010}, and to estimate the amount of generated randomness \cite{Massar2012,Vertesi2014,Woodhead2020}. 

The maximum exponent of the certification power of Bell inequalities is known as self-testing. Introduced in \cite{mayers2003self}, self-testing allows for almost complete characterization of the underlying quantum system based only on the observed Bell violation. It thus appears to be one of the most accurate methods
for certification of quantum systems which makes self-testing a highly valuable asset for the rapidly developing of quantum technologies. In fact, self-testing techniques have shown to be amenable for near-term quantum devices, allowing for a proof-of-principle state certification of up to few tens of particles \cite{we2021robust,yang2022testing}. For this reason self-testing has attracted a considerable attention in recent years (see, e.g., Ref. \cite{Supic2020selftestingof}).

However, most of the above applications require Bell inequalities that exhibit carefully crafted features. In the particular case of self-testing one needs Bell inequalities whose maximal quantum values are achieved by the target quantum state and measurements that one aims to certify. Deriving Bell inequalities tailored to generic pure entangled states turns out to be in general a difficult challenge. Even more so if one looks for inequalities applicable to systems of arbitrary number of parties or arbitrary local dimension.
The standard geometric approach to derive Bell inequalities has been successful in deriving many interesting  and relevant inequalities \cite{CHSH,CGLMP,BKP,LPZB,ZB}, but
typically fails to serve a self-testing purpose, providing inequalities with unknown maximal quantum violation.

In order to construct Bell inequalities that are tailored to specific quantum states, a more promising path is to exploit the "quantum properties" of the considered system such as its symmetries. Two proposals in this direction have succeeded in designing different classes of Bell inequalities tailored to the broad family of multi-qubit graph states \cite{PhysRevLett.95.120405,BellGS2006} and the first Bell inequalities maximally violated by the maximally entangled state of any local dimension \cite{SalavrakosPRL}. 
The success of these methods was further confirmed by later applications to design the first self-testing Bell inequalities for graph states \cite{PhysRevLett.95.120405}, for genuinely entangled stabilizer subspaces \cite{Subspaces1,Subspaces2} or maximally entangled two-qutrit states \cite{kaniewski2018maximal}, as well as to derive many other classes of Bell inequalities tailored to two-qudit maximally entangled \cite{JKSA,KaniewskiMeasurements} or many-qudit Greenberger-Horne-Zeilinger states \cite{Augusiak2019}. Some of these constructions were later exploited to provide self-testing schemes for the maximally entangled  \cite{JKSA,sarkar2019self} or the GHZ states \cite{SarkarGHZ} of arbitrary local dimension.

In this work we show that similar ideas can be used to provide a general construction of Bell inequalities tailored to graph states of arbitrary prime local dimension. Graph states constitute one of the most representative classes of genuinely entangled multipartite quantum states considered in quantum information, covering the well-known Greenberger-Horne-Zeilinger, the cluster \cite{cluster} or the absolutely maximally entangled states \cite{helwig2013absolutely}, that have found numerous applications, e.g., in quantum computing \cite{Yao2012-wq,PhysRevLett.86.5188,Briegel_2009} or quantum metrology \cite{T_th_2014}. Interestingly, our construction provides the first example of Bell inequalities maximally violated by the absolutely maximally entangled states of non-qubit local dimension such as the four-qutrit AME(4,3) state \cite{helwig2013absolutely}. Moreover, it generalizes and unifies in a way the constructions of Refs. \cite{baccari2018scalable} and \cite{kaniewski2018maximal} to graph states of arbitrary prime local dimension.

The manuscript is organized as follows. In Sec. \ref{Sec:Preliminaries} we provide some background information which is necessary for further considerations; in particular we explain in detail the notions of the multipartite Bell scenario and graph states and also state the definition of self-testing we use in our work. Next, in Sec. \ref{Sec:Construction} we introduce our general construction of Bell inequalities for graph states. We then show in Sec. \ref{Sec:Self-testing} that our new Bell inequalities allow for self-testing of all graph states of local dimension three. We conclude in Sec. \ref{Sec:Conclusion} where we also provide a list of possible research directions for further studies that follow from our work.

\section{Preliminaries}
\label{Sec:Preliminaries}

\subsection{Bell scenario and Bell inequalities}
\label{BellSc}

Let us begin by introducing some notions and terminology. 
We consider a multipartite Bell scenario in which $N$ distant observers $A_i$
share a quantum state $\rho$ defined on the product Hilbert space 
\begin{equation}
    \mathcal{H}=\mathcal{H}_1\otimes\ldots\otimes\mathcal{H}_N.
\end{equation}
Each observer $A_i$ 
can perform one of $m_i$ measurements $M_{x_i}^i\equiv\{M_{a_i|x_i}^{i}\}_{a_i}$ on their share of this state, where $x_i$ stand for the measurement choices, whereas $a_i$ denote the outcomes; here we label them as $x_i=1,\ldots,m$ and $a_i=0,\ldots,d-1$, respectively. Recall that the measurement operators satisfy $M_{a_i|x_i}^i\geq 0$ for any choice of $a_i$ and $x_i$ as well as $\sum_{a_i}M_{a_i|x_i}^{i}=\mathbbm{1}$ for any $x_i$.

The observers repeat their measurements on the local parts of the state $\rho$ which creates correlations between the obtained outcomes. These are captured by a collection of probability distributions $\vec{p}\equiv \{p(\va|\vx)\}\in\mathbbm{R}^{(md)^N}$, where $p(\va|\vx)\equiv p(a_1,\ldots,a_N|x_1,\ldots,x_N)$ is the probability of obtaining the outcome $a_i$ by the observer $i$ upon performing the measurement $M_{x_i}^i$ and can be represented by the Born rule
\be \label{jp1}
p(\va|\vx) =  \Tr\left[\rho\left(M_{a_1|x_1}^1\otimes \ldots \otimes M_{a_N|x_N}^N\right)\right].
\ee 

A behaviour $\vec{p}$ is said to be local or classical if for any $\vec{a}$ and $\vec{x}$, the joint probabilities $p(\vec{a}|\vec{x})$ factorize in the following sense, 
\be \label{local}
p(\va|\vx) = \sum_{\lambda} \mu(\lambda) p_1(a_1|x_1,\lambda)\cdot \ldots \cdot p_N(a_N|x_N,\lambda),
\ee
where $\lambda$ is a random variable with a probability distribution $\mu(\lambda)$ representing the possibilities for the parties to share classical correlations
and $p_i(a_i|x_i,\lambda)$ is an arbitrary probability distribution corresponding to the observer $A_i$. On the other hand, if a behavior $\vec{p}$ does not admit the above form, we call it Bell non-local or simply non-local. In any Bell scenario correlations that are classical in the above sense form a polytope with finite number of vertices, denoted $L_{N,m,d}$. 

Any non-local distribution $\vec{p}$ can be detected to be outside the local polytope from the violation of a Bell inequality. The generic form of such inequalities is
\begin{equation}
    I:=\sum_{\va,\vx}\alpha_{\va,\vx}\,p(\va|\vx)\leq \beta_L,
\end{equation}
where $\beta_L=\max_{\vec{p}\in L_{N,m,d}}I$ is the classical bound of the inequality
and $\alpha_{\va,\vx}$ are some real coefficients defining the inequality. Any $\vec{p}$ that violates a Bell inequality is detected as non-local. 

Let us finally introduce another number characterizing a Bell inequality---the so-called quantum or Tsirelson's bound---which is defined as
\begin{equation}
    \beta_Q=\sup_{\vec{p}\in Q_{N,m,d}}I,
\end{equation}
where the maximisation runs on all quantum behaviours, i.e., all distributions $\vec{p}$ that can be obtained by performing quantum measurements on quantum states of arbitrary local dimension. The set of quantum correlations $Q_{N,m,d}$ is in general not closed \cite{Slofstra} and thus $\beta_Q$ is a supremum and not a strict maximum. Determining the quantum bound for a generic Bell inequality is an extremely difficult problem. However, interestingly, in certain cases it can still be found analytically. A way to obtain $\beta_Q$ or at least an upper bound on it is to find a sum-of-squares decomposition of a Bell operator $\mathcal{B}$ corresponding to the Bell inequality. More specifically, if for any choice of measurement operators one is able to represent the Bell operator as
\be \label{sos-m}
\B = \eta \I - \sum_k P^\dagger_k P_k,
\ee 
where $P_k$ are some operators composed of $M_{n_i|x_i}^i$, then 
$\eta$ is an upper bound on $\beta_Q$. Indeed, Eq. (\ref{sos-m}) implies that for all $|\psi\rangle$, $\langle \psi|\B|\psi\rangle \leqslant \eta$, and thus, $\beta_Q\leq \eta$. If a quantum state saturates this upper bound, then it follows from \eqref{sos-m} that $P_k|\psi\rangle = 0$ for all $k$. As we will see later such relations are particularly useful to prove a self-testing statement from the maximal violation of a Bell inequality.

%

For further convenience we also introduce an alternative description of the Bell scenario in terms of generalized expectation values (see, e.g., Ref. \cite{Augusiak2019}). These are in general complex numbers defined through the $N$-dimensional discrete Fourier transform of $\{p(\va|\vx)\}$, 
\be \label{jp}
\langle A_{n_1|x_1}^1 \ldots A_{n_N|x_N}^N \rangle=  \sum_{\va} \w^{\va \cdot \Vec{n}} p(\va|\vx) ,
\ee 
where $\omega=\mathrm{exp}(2\pi \mathrm{i}/d)$ is the $d$th root of unity, 
$\Vec{a} := (a_1,\ldots,a_N) \in \{0,\ldots,d-1\}^N$ and $\Vec{n} := (n_1,\ldots,n_N) \in \{0,\ldots,d-1\}^N$, and $\va \cdot \Vec{n}=\sum_ia_in_i$. The inverse transformation gives
\be \label{jp}
p(\va|\vx)  = \frac{1}{d^N} \sum_{\Vec{n}} \w^{- \va \cdot \Vec{n}} \langle A_{n_1|x_1}^1 \ldots A_{n_N|x_N}^N \rangle.
\ee 
Combining Eqs. (\ref{jp1}) and (\ref{jp}) one finds that if the correlations $\vec{p}$ are quantum, that is, originate from performing local measurements on composite quantum states, the complex expectation values can be represented as
\begin{equation}
\langle A_{n_1|x_1}^1 \ldots A_{n_N|x_N}^N \rangle=\Tr\left[\rho\left(A_{n_1|x_1}^1 \otimes\ldots\otimes A_{n_N|x_N}^N\right)\right],    
\end{equation}
where $A_{n_i|x_i}^i$ are simply Fourier transforms of the 
measurement operators $M_{a_i|x_i}^i$ given by
\be \label{ft}
A_{n_i|x_i}^i = \sum_{a_i=0}^{d-1} \w^{n_i a_i} M_{a_i|x_i}^i.
\ee 
Clearly, due to the fact that the Fourier transform is invertible, for a given $x_i$ and $i$, the $d$ operators $A_{n_i|x_i}^i$ with $n_i=0,\ldots,d-1$
uniquely represent the corresponding measurement $M_{x_i}^i$. 

Let us now discuss a few properties of the Fourier-transformed measurement operators that will prove very useful later. For clarity of the presentation we consider a single quantum measurement $M=\{M_a\}$ and the corresponding $A_n$ operators obtained via Eq. (\ref{ft}). First, one easily finds that $A_0=\mathbbm{1}$. Second, 
\begin{equation}
    A_{d-n}=A_{-n}=A_{n}^{\dagger}
\end{equation}
which is a consequence of the fact that $\omega^{d-n}=\omega^{-n}=(\omega^{n})^*$ holds true for any $n\in\{0,\ldots,d-1\}$. Third, $A_n^{\dagger}A_n\leq \mathbbm{1}$ for any $n=0,\ldots,d-1$ (for a proof see Ref. \cite{kaniewski2018maximal}). 

Let us finally mention that if $M$ is projective then all $A_n$ are unitary and their eigenvalues are simply powers of $\w$; equivalently $A_n^d=\mathbbm{1}$. It is also not difficult to see that in such a case, $A_{n}$ are operator powers of $A_{1}$, that is, $A_{n}=A_{1}^n$. In this way we recover a known fact that a projective measurement can be represented by a single observable. We exploit these properties later in our construction of Bell inequalities as well as in deriving the self-testing statement. In fact, in what follows we denote the observables measured by the party $i$ by $A_{i,x_i}$.

\subsection{Self-testing}

Here we introduce the definition of $N$-partite self-testing that we adopt in this work. Let us consider again the Bell scenario described above, assuming, however, that the shared state $\rho$, the Hilbert space it acts on as well as the local measurements are all unknown. The aim of the parties is to deduce their form from the observed correlations $p(\vec{a}|\vec{x})$. Since the dimension of the joint Hilbert space $\mathcal{H}$ is now unconstrained (although finite) we can simplify the latter problem by assuming that the shared state is pure, i.e., $\rho=\proj{\psi}$ for some $\ket{\psi}\in\mathcal{H}$, and the measurements are projective, in which case they are represented by unitary observables $A_{i,x_i}$ acting on $\mathcal{H}_i$.

Consider then a target state $\ket{\hat{\psi}}\in (\mathbb{C}^d)^{\otimes N}$ and the corresponding measurements $\hat{A}_{i,x_i}$, giving rise to the same behaviour $\{p(\vec{a}|\vec{x})\}$. We say that the observed correlations self-test the given state and measurements if the following definition applies.

\begin{defn}\label{definition} 
If from the observed correlations $\{p(\vec{a}|\vec{x})\}$ one can identify a qudit in each local Hilbert space in the sense that $\mathcal{H}_i=\mathbbm{C}^d\otimes \mathcal{H}_i'$ for some auxiliary Hilbert space $\mathcal{H}_i'$, 
and also deduce the existence of local unitary operations 
$U_i:\mathcal{H}_i\to\mathbbm{C}^d\otimes \mathcal{H}_i'$ such that 
\begin{equation}
    (U_1\otimes\ldots\otimes U_N)\ket{\psi}=\ket{\hat{\psi}}\otimes \ket{\mathrm{aux}}
\end{equation}
for some $\ket{\mathrm{aux}}\in\mathcal{H}_i'\otimes\ldots\otimes\mathcal{H}_{N}'$, 
and, moreover,
\begin{equation}
    U_i\,A_{i,x_i}\, U_i^{\dagger}=\hat{A}_{i,x_i} \otimes \I_{i},
\end{equation}
where $\I_i$ is the identity acting on $\mathcal{H}'_{i}$, 
then we say that the reference quantum state $\ket{\hat{\psi}}$
and measurements $\hat{A}_{i,x_i}$ have been self-tested in the experiment.
\end{defn}

Importantly, only non-local correlations can give rise to a valid self-testing statement. 
Moreover, since it is based only on the observed correlations $\{p(\vec{a}|\vec{x})\}$, self-testing can characterize the state and the measurements only up to certain equivalences. In particular, the statement above includes all possible operations that keep the correlations $\{p(\vec{a}|\vec{x})\}$ unchanged, such as: (i) the addition of an auxiliary state $\ket{\mathrm{aux}}$ on which the measurements act trivially and (ii) the rotation by an arbitrary local unitary operations.
It is also not difficult to check that 
$\{p(\vec{a}|\vec{x})\}$ does not change if one applies the transposition map to the quantum state as well as all the observables, and thus sometimes one needs to take into account this extra degree of freedom, which leads to a slightly weaker definition of self-testing \cite{kaniewski2018maximal,KaniewskiWeak}. In the present case, since the graph states are all real the transposition does not pose any problem as far as the state self-testing is concerned, yet it does for the measurements.

\subsection{Graph states}\label{sec: graph states}

Let us finally recall the definition of multipartite graph states of prime local dimension \cite{PhysRevA.65.012308,PhysRevA.71.042315,hein2006entanglement}.
Consider a graph $\mathcal{G}=(\mathcal{V},\mathcal{E},\mathcal{R},d)$, where $d$ is any prime number such that $d\geq 2$, $\mathcal{V} := \{1,\ldots,N\}$ is the set of vertices of the graph, $\mathcal{E}$ is the set of edges connecting vertices, and $\mathcal{R}:= \{r_{i,j}\}$ is a set of natural numbers from $\{0,\ldots,d-1\}$ specifying the number of edges connecting vertices $i,j \in \mathcal{V}$; in particular, $r_{i,j}=0$ means there is no edge between $i$ and $j$. We additionally assume that $r_{i,i}=0$ for all $i$, meaning that the graph has no loops as well as that the graph $\mathcal{G}$ is connected, meaning that it does not have any isolated vertices. 
By $\mathcal{N}_i$ we denote the neighbourhood of the vertex $i$ which consists
of all elements of $\mathcal{V}$ that are connected to $i$.

Assume then that each vertex $i\in\mathcal{V}$ of the graph corresponds to a single quantum system held by the party $A_i$ and let us associate to it the following $N$-qudit operator 
\begin{equation} \label{Gi-m}
    G_{i} = X_{i}  \otimes \bigotimes_{j\in \mathcal{N}_i} Z_{j}^{r_{ij}}\qquad (i=1,\ldots,N)
\end{equation} 
with $X$ and $Z$ being the generalizations of the qubit Pauli matrices to $d$-dimensional Hilbert spaces defined via the following relations 
\begin{equation}
    Z\ket{i}=\omega^i\ket{i},\qquad X\ket{i}=\ket{i+1}\qquad (i=0,\ldots,d-1),
\end{equation}
where the addition is modulo $d$. Due to the fact that $XZ=\omega^{-1} ZX$, it is not difficult to see that 
the operators $G_i$ mutually commute. It then follows that there is a unique
pure state $\ket{G}\in(\mathbbm{C}^d)^{\otimes N}$, called graph state, which is a common eigenstate of all $G_i$ corresponding to the eigenvalue one, i.e., 
\begin{equation}
    G_i\ket{G}=\ket{G}\qquad (i=1,\ldots,N).
\end{equation}
Given the above property, the $G_i$ are usually referred to as stabilizing operators. Notice also that 
in the particular case of $d=2$ this construction naturally reproduces 
the $N$-qubit graph states \cite{hein2006entanglement}, where vertices 
can only be connected by single edges.

Let us illustrate the above construction with a couple of examples. 

\paragraph{Example 1: Maximally entangled two-qudit state.} Let us start with the simplest possible graph, consisting of two vertices connected by an edge (cf. Fig. \ref{fig:ame}(a)).
The corresponding generators are given by
\begin{equation}
G_1=X\otimes Z,\qquad    G_2=Z\otimes X,
\end{equation}
and stabilize a single state in $\mathbbm{C}^d\otimes \mathbbm{C}^d$ which is equivalent up to local unitary operations to the maximally entangled state of two qudits,  
\begin{equation}\label{maxent}
    |\psi_d^+\rangle=\frac{1}{\sqrt{d}}\sum_{i=0}^{d-1}|ii\rangle
\end{equation}
in which both local Schmidt bases are the computational one. 
In fact, the above state is stabilized by another pair of operators, namely,
\begin{equation}\label{maxentstab}
    G'_1=X\otimes X,\qquad G'_2=Z\otimes Z^{\dagger},
\end{equation}
which are obtained from $G_i$ by an application of 
the Fourier matrix to the second site.

\paragraph{Example 2: GHZ state.} The above two-vertex graph naturally generalizes to a star graph consisting of $N$ vertices (cf. Fig. \ref{fig:ame}(b)). The associated generators are of the form 
\begin{equation}
    G_1=X_1Z_2\ldots Z_N
\end{equation}
and
\begin{equation}
    G_i=Z_1X_i\qquad (i=2,\ldots,N),
\end{equation}
%
and stabilize an $N$-qudit state which is equivalent under local unitary operations to the well-known Greenberger-Horne-Zeilinger (GHZ) state 
\begin{equation}
    \ket{\mathrm{GHZ}_{N,d}}=\frac{1}{\sqrt{d}}\sum_{i=0}^{d-1}\ket{i}^{\otimes N}.
\end{equation}

\begin{figure}[http]
    \centering
    \includegraphics[width=\columnwidth]{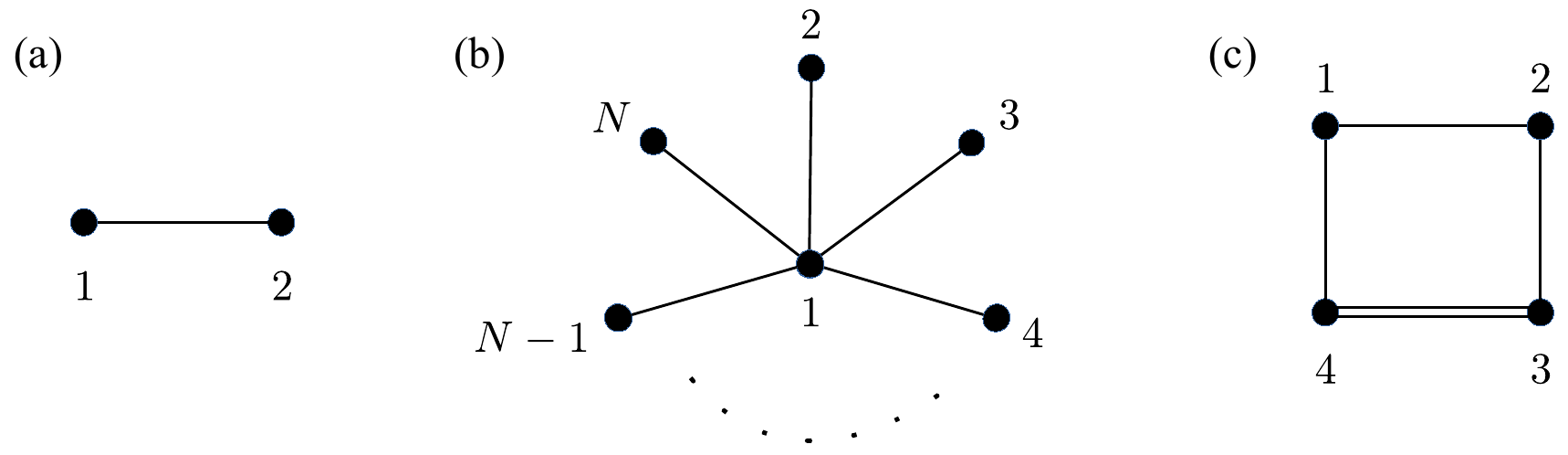}
    \caption{Three examples of graphs defining: (a) the maximally entangled state of two qudits, (b) the $N$-qudit GHZ state, (c) the four-qutrit absolutely maximally entangled state AME(4,3).}
    \label{fig:ame}
\end{figure}

\paragraph{Example 3: AME(4,3).} The third and the last example is concerned with the four-qutrit absolutely maximally entangled state\footnote{A multipartite state is termed absolutely maximally entangled if any of its $\lfloor N/2\rfloor$-partite subsystems is in the maximally mixed state \cite{PhysRevA.77.060304}.}, named AME(4,3) \cite{helwig2013absolutely}. The graph defining it is presented in Fig. \ref{fig:ame}(c).
The stabilizing operators corresponding to this graph 
read
\begin{equation}\label{Stab:AME}
    G_1=X_1Z_2Z_4,\qquad G_2=Z_1X_2Z_3,\qquad G_3=Z_2X_3Z_4^2,\qquad
    G_4=Z_1Z_3^2X_4.
\end{equation}
They stabilize a three-qutrit maximally entangled state
AME(4,3) which is equivalent under local unitary operations and relabelling of the subsystems to (see, e.g., Ref. \cite{PhysRevA.100.022342}),
\begin{equation}
    \ket{\mathrm{AME}(4,3)}=\frac{1}{3}\sum_{i,j=0}^2
    \ket{i}\ket{j}\ket{i+j}\ket{i+2j},
\end{equation}
where the addition is modulo three.

\section{Construction of Bell inequalities for arbitrary graph states of prime local dimension}
\label{Sec:Construction}

Here we present our first main result: a general construction of Bell inequalities whose maximal quantum value is achieved by the $N$-qudit graph states of arbitrary prime local dimension and quantum observables corresponding to mutually unbiased bases at every site. Our construction is inspired by the recent approach to construct CHSH-like Bell inequalities for the $N$-qubit graph states presented in Ref. \cite{baccari2018scalable} and by another construction of Bell inequalities maximally violated by the maximally entangled two-qudit state introduced in Ref. \cite{kaniewski2018maximal}. 

First, in Sec. \ref{sec:qubit graph states} we recall the general class of Bell inequalities maximally violated by $N$-qubit graph states of Ref. \cite{baccari2018scalable}. Then, in Sec. \ref{sec:replacement} we introduce the main building block to generalise this construction to arbitrary prime dimension. We illustrate the Bell inequality construction with some simple examples in Sec. \ref{sec:examples} and then move to introduce the general form of the inequality valid of any $N$-qudit graph state of prime dimension in Sec. \ref{sec: general construction}.

\subsection{Multiqubit graph states}\label{sec:qubit graph states}

Let us assume that $d=2$ and let us consider a graph $\mathcal{G}$. Without any loss of generality we can assume that a vertex with the largest 
neighbourhood is the first one, that is, $N_1=\max_{i=1,\ldots,N}|\mathcal{N}_i|$. If there are many vertices with the maximal neighbourhood in $\mathcal{G}$, we are free to choose any of them as the first one.

To every generator $G_i$ we associate an expectation value in which the $X$ and $Z$ Pauli matrices are replaced by quantum observables or their combinations using the following rule. At the first qubit we make the following assignment,
\begin{equation}\label{replacement}
X\to \frac{1}{\sqrt{2}}(A_{1,0}+A_{1,1}),\qquad 
Z\to \frac{1}{\sqrt{2}}(A_{0,1}-A_{1,1}),
\end{equation}
whereas the Pauli matrices at the remaining sites are directly replaced by observables, that is,
\begin{equation}
    X\to A_{i,0},\qquad Z\to A_{i,1}
\end{equation}
with $i=2,\ldots,N$. Recall that the first index enumerates the parties, while the second one measurement choices. This procedure gives us $N$ expectation values which
after being combined altogether lead us to the following Bell inequality 
\cite{baccari2018scalable}:
\begin{eqnarray}\label{QubitBellIneq}
    I_G&:=&\frac{N_1}{\sqrt{2}}\left\langle(A_{1,0}+A_{1,1})\prod_{i\in \mathcal{N}(1)}A_{i,1}\right\rangle+\frac{1}{\sqrt{2}}\sum_{i\in \mathcal{N}(1)}\left\langle(A_{1,0}-A_{1,1})A_{i,0}\prod_{j\in \mathcal{N}(1)\setminus\{1\}}A_{j,1}\right\rangle
    \nonumber\\
    &&+\sum_{i\notin \mathcal{N}(1)\cup\{1\}}\left\langle A_{i,0}\prod_{j\in \mathcal{N}(i)}A_{j,1}\right\rangle\leq \beta_{C}^G,
\end{eqnarray}
where the classical bound can directly be determined for any graph $G$ and is given by $\beta_C^G=N+(\sqrt{2}-1)N_{1}-1$. More importantly, the maximal quantum value can also be analytically 
computed for any graph and amounts to $\beta_Q^G=N+N_{1}-1$.
This value is achieved by the graph state $\ket{G}\in (\mathbbm{C}^2)^{\otimes N}$ corresponding to the graph $\mathcal{G}$ and the following observables:
\begin{equation}\label{QubitComb}
    A_{1,0}=\frac{1}{\sqrt{2}}(X+Z),\qquad 
    A_{1,1}=\frac{1}{\sqrt{2}}(X-Z)
\end{equation}
for the first observer and $A_{i,0}=X$ and $A_{i,1}=Z$
for the remaining observers $i=2,\ldots,N$.

It is worth stressing here that one of the key observations making the construction of Ref. \cite{baccari2018scalable} work is that for any graph there exists a choice of observables at any site, given by the above formulas, turning the quantum operators appearing in 
the expectation values of (\ref{QubitBellIneq})
into the stabilizing operators $G_i$; in particular, 
it is a well-known fact that combinations of the Pauli matrices in Eq. (\ref{QubitComb}) are proper quantum observables with eigenvalues $\pm 1$.

\subsection{Replacement rule for operators of arbitrary prime dimension}\label{sec:replacement}

We now move on to introduce the main ingredient needed to  generalise the above construction to graph states of prime local dimension $d \geq 3$.

A naive approach to constructing Bell inequalities for 
graph states of higher local dimensions would be to directly 
follow the $d = 2$ strategy. That is, at a chosen site the $X$ and $Z$ operators
are replaced by combinations of general $d$-outcome observables $A_{0}$ and $A_{1}$.
However, this simple approach fails to work beyond $d = 3$ because for any prime $d\geq 3$ it is 
impossible find nonzero complex numbers $\alpha,\beta\in\mathbbm{C}$ for which 
%
%
\begin{equation}
    O=\alpha X+\beta Z,
\end{equation}
is a valid quantum observable; in fact, for no complex numbers the above combinations can be unitary, unless $d=2$ (cf. Fact \ref{fact:unitary} in Appendix \ref{App0}).
This makes the transformation (\ref{QubitComb}) irreversible.
%
%
Phrasing differently, there are no unitary observables $A_0$ and $A_1$ such that $X=\alpha A_0+\beta A_1$ and $Z=\delta A_0+\gamma A_1$ for some complex numbers $\alpha,\beta,\gamma,\delta\in\mathbbm{C}$. 

Nevertheless, there exist other sets of $d$-outcome quantum observables which can be linearly combined to form quantum observables, and thus are convenient for our purposes. One such choice is the following set of $d$ unitary matrices
\begin{equation}\label{Ok}
   O_{k}:= XZ^k\qquad (k=0,\ldots,d-1).
\end{equation}
It is not difficult to check that $O_{k}^d=\mathbbm{1}_d$ for any $k=0,\ldots,d-1$ and prime $d$, meaning that the eigenvalues of each of these unitary matrices belong to the set $\{1,\omega,\ldots,\omega^{d-1}\}$, and thus 
are proper $d$-outcome observables in our formalism. It is also worth mentioning that for any prime $d\geq 2$ their eigenvectors together with the standard basis in $\mathbbm{C}^d$ form $d+1$ mutually unbiased bases.

Let us now assume that $d$ is a prime number greater than two ($d\geq 3$) 
and consider the following linear combinations of $O_{k}$ and their powers, 
\begin{eqnarray} \label{At-m}
\overline{O}_{x}^{(n)} = \frac{\lambda_n}{\sqrt{d}} \sum^{d-1}_{k=0} \w^{nxk} 
\omega^{nk(k+1)} O_{k}^n,
\end{eqnarray}
where $x=0,1,\ldots,d-1$ and $\lambda_n$ are complex coefficients
defined as \cite{kaniewski2018maximal}:
\begin{equation}\label{lambdan}
    \lambda_n=\left[\varepsilon_d\left(\frac{n}{d}\right)\right]^{-1}\omega^{-g(n,d)/48},
\end{equation}
where 
\begin{equation}\label{lambdan2}
    \varepsilon_d:=
    \left\{
    \begin{array}{ll}
    1, &\quad \mathrm{if}\quad d\equiv1\mod4,\\[1ex]
    \mathbbm{i}, &\quad \mathrm{if}\quad d\equiv3\mod4.
    \end{array}
    \right.
\end{equation}
$\left(\frac{n}{d}\right)$ is the Legendre symbol\footnote{Recall that the Legendre symbol $\left(\frac{n}{d}\right)$  equals $+1$ if $n$ is a
quadratic residue modulo $d$ and $-1$ otherwise}, and, finally, the 
coefficients $g(n,d)$ are given by 
\be\label{lambdan3}
g(n,d) = \left \{ 
\begin{array}{lll}
n[n^2-d(d+6)+3] & \mbox{if} & n\equiv 0 \mbox{ mod 2 and } n+d+1/2 \equiv 0 \mbox{ mod 2,} \\ [1ex]
n[n^2-d(d-6)+3] & \mbox{if} & n\equiv 0 \mbox{ mod 2 and } n+d+1/2 \equiv 1 \mbox{ mod 2,}  \\ [1ex]
n(n^2+3)+2d^2(-5n+3) & \mbox{if} & n\equiv 1 \mbox{ mod 4,} \\[1ex] 
n(n^2+3)+2d^2(n+3) & \mbox{if} & n\equiv 3 \mbox{ mod 4.} 
\end{array} 
\right.
\ee
%
Importantly, it was proven in Ref. \cite{kaniewski2018maximal} (see Appendix D therein) that $\overline{O}_{x}^{(n)}$ are unitary and satisfy 
\begin{equation}\label{O_n^d}
    \left[\overline{O}_{x}^{(n)}\right]^d=\mathbbm{1}_d
\end{equation}
for any $x=0,\ldots,d-1$ and $n=1,\ldots,d-1$. What is more,
$\overline{O}_{x}^{(n)}$ turns out to be the $n$th power of 
$\overline{O}_{x}$, that is,
$\overline{O}_{x}^{(n)}=[\overline{O}_{x}]^{n}$.
All this means that for any $x$ the set $\{\overline{O}_{x}^{(n)}\}_{n=0,\ldots,d-1}$ represents a legitimate $d$-outcome projective quantum measurement. 
Let us finally mention that the linear transformation (\ref{At-m}) 
can be inverted, giving
\be \label{Ak-m}
O^n_{l} = \frac{\omega^{-nl(l+1)}}{\sqrt{d}\,\lambda_n} \sum^{d-1}_{x=0} \w^{-nxl}  \,
\overline{O}_{x}^{(n)}.
\ee

The fact that both $O_k$ and $\overline{O}_{k}$ are unitary quantum observables that are related by a linear reversible transformation given by Eqs. (\ref{At-m}) and (\ref{Ak-m}) is the key ingredient in our construction. That is, we can proceed in analogy to $d=2$ case, where we used the replacement defined in Eq. (\ref{replacement}) to define the Bell inequality and we could later reverse it by a suitable choice of quantum observables (\ref{QubitComb}) to obtain the maximal quantum violation with a graph state. 

The replacement rule we use for the case of arbitrary prime dimension becomes:
\begin{equation}\label{assign1}
    \left(XZ^k\right)^n\,\to\, \widetilde{A}^{(n)}_k:= \frac{\omega^{-nk(k+1)}}{\sqrt{d}\,\lambda_n} \sum^{d-1}_{t=0} \w^{-ntk}  A_{t}^n,
\end{equation}
where $A_{t}$ with $t=0,\ldots,d-1$ are unitary observables.
Notice that since we deal now
with $d$-outcome quantum measurements we need to also take into account the powers 
$n$ of the corresponding observables. In fact, these under the Fourier transform represent the outcomes of projective measurements.

Crucially, this transformation can be inverted in the sense that there exist a choice of observables $A_{i,t}$,
\begin{equation}\label{AObs}
    A_{t}^n=\frac{\lambda_n}{\sqrt{d}}\sum_{k=0}^{d-1}\omega^{ntk}
    \omega^{nk(k+1)}\left(XZ^k\right)^n.
\end{equation}
for which $\widetilde{A}_k^{(n)}$ in Eq. (\ref{assign1}) can be brought back to $XZ^k$.

These new operators $\widetilde{A}^{(n)}_k$ satisfy the following relations (see Fact \ref{fact:properties} in Appendix \ref{App0} for a proof): 
\begin{equation}\label{prop1}
\left(\widetilde{A}^{(n)}_k\right)^\dagger = \widetilde{A}^{(d-n)}_k = \widetilde{A}^{(-n)}_k    
\end{equation}
for any pair $n,k=0,\ldots,d-1$,
and
\begin{equation}\label{prop2}
\sum^{d-1}_{k=0}  \widetilde{A}^{(d-n)}_k \widetilde{A}^{(n)}_k = d \I
\end{equation}
for any $n=0,\ldots,d-1$.

\subsection{Examples}\label{sec:examples}

Before presenting our construction in full generality, let us first illustrate how to use the qudit replacement rule to obtain valid Bell inequalities tailored to graph states by means of two examples.

\paragraph{Example 1: AME(4,3).} As mentioned in Sec. \ref{sec: graph states}, the 
the four-qutrit absolutely maximally entangled state is a graph state corresponding to the graph presented on Fig. \ref{fig:ame}.
The stabilizing operators 
defining this state are given in Eq. (\ref{Stab:AME}).
We recall them here
\begin{equation}
    G_1=X_1Z_2Z_4,\qquad G_2=Z_1X_2Z_3,\qquad G_3=Z_2X_3Z_4^2,\qquad
    G_4=Z_1Z_3^2X_4.
\end{equation}

Since the neighbourhood of all vertices of this graph is of size two, each vertex is equally good  to implement the transformation (\ref{assign1}). For simplicity 
we choose it to be the first site. Moreover, as in the previous example, we denote the 
observables measured by the four parties as $A_x$, $B_y$, etc.

Now, to create the set of matrices $XZ^k$ [necessary for the transformation (\ref{assign1})] at the first site we consider 
the stabilizing operators $G_1$, $G_1G_2$, and $G_1G_2^2$. These are, however, insufficient to uniquely define $\ket{\mathrm{AME}(4,3)}$ as they do not include $G_3$ and $G_4$. Since $G_3$ has the identity at the first position we can include it as it is, whereas we need to take a product of $G_4$ with $G_1$ to create $XZ$ at the first site. As a result, the final set of stabilizing operators which we use to construct a Bell inequality for $\ket{\mathrm{AME}(4,3)}$ consists of 
\begin{eqnarray}\label{AMEStab}
    G_1&\!\!\!=\!\!\!&X\otimes Z\otimes\mathbbm{1}\otimes Z,\nonumber\\
    G_1G_2&\!\!\!=\!\!\!&XZ\otimes ZX\otimes Z \otimes Z,\nonumber\\
    G_1G_2^2&\!\!\!=\!\!\!&XZ^2\otimes ZX^2\otimes Z^2\otimes Z,\nonumber\\
    G_3&\!\!\!=\!\!\!&\mathbbm{1}\otimes Z\otimes X\otimes Z^2\nonumber\\
    G_1G_4&\!\!\!=\!\!\!&XZ\otimes Z\otimes Z^2\otimes ZX.
\end{eqnarray}

Now, to each of these stabilizing operators we associate 
an expectation value in which particular matrices are 
replaced by quantum observables or their combinations. 
For pedagogical purposes, let us do it site by site. 
As already mentioned, at the first site we use Eq. (\ref{assign1})
which for $d=3$ gives
\begin{eqnarray}\label{assignd3}
    X&\to& \widetilde{A}_0:=\frac{1}{\sqrt{3}\,\lambda_1}\left(A_0+A_1+A_2\right),\nonumber\\
    XZ&\to& \widetilde{A}_1:=\frac{1}{\sqrt{3}\,\lambda_1\omega}\left(A_0+\omega^{-1} A_1+\omega^{-2}A_2\right),\nonumber\\
    XZ^2&\to& \widetilde{A}_2:=\frac{1}{\sqrt{3}\,\lambda_1}\left(A_0+\omega^{-2} A_1+\omega^{-1} A_2\right),
\end{eqnarray}
where $\lambda_1=-\mathrm{i}\omega^{2/3}=\omega^{1/12}=\mathrm{exp}(\pi\mathbbm{i}/18)$ and $\lambda_2=\lambda_1^*$ [cf. Eqs. (\ref{lambdan}), (\ref{lambdan2}) and (\ref{lambdan3})] and we denoted for simplicity $\widetilde{A}_i\equiv\widetilde{A}_i^{(1)}$. We dropped the subscript $n$ appearing in the transformation (\ref{assign1}) because for $n=2$ one has
$(XZ^k)^2=(XZ^k)^{\dagger}$ for $k=0,1,2$ and 
$\widetilde{A}_i^{(2)}=\widetilde{A}_i^{\dagger}$ [cf. Eq. (\ref{prop1})]; nevertheless, we need to take into account the case $n=2$ when constructing the Bell inequality.

We then note that at the second site we also have three independent 
unitary observables [note that $(ZX)^{3}=(ZX^2)^3=\mathbbm{1}$] and therefore we can directly substitute
\begin{equation}\label{assign5}
    Z\to B_0,\qquad ZX\to B_1,\qquad ZX^2\to B_2.
\end{equation}
At the third site we have $Z$, $Z^2$ which represent a single
measurement (cf. Sec. \ref{BellSc}), and $X$ which is independent 
of the other two. We thus substitute
%
 $   Z^k\to C_0^k$ with $k=1,2$ and $X\to C_1$.
%
Analogously, for the fourth party we have 
%
 $   Z\to D_0$ and $ZX\to D_1.$
%

Taking all the above substitutions into account we arrive at the following assignments
\begin{eqnarray}\label{transformation}
    G_1&\to&\langle\widetilde{A}^{(1)}_0B_0D_0\rangle, \qquad
    G_1G_2\to\langle\widetilde{A}_1^{(1)}B_1C_0D_0\rangle,\qquad
    G_1G_2^2\to\langle\widetilde{A}_2^{(1)}B_2C_0^2D_0\rangle,
\end{eqnarray}
\begin{equation}
    G_1G_4\to\langle\widetilde{A}_1^{(1)}B_0C_0^2D_1\rangle,
\end{equation}
%
%
and for $G_3$:
\begin{equation}\label{transformation2}
      G_3\to \langle B_0 C_1D_0\rangle.
\end{equation}
Notice that the expectation values corresponding to $n=2$ in the assignment 
(\ref{assign1}) are simply complex conjugations of the above ones.
By adding all the obtained expectation values, 
we finally obtain a Bell inequality of the form
\begin{eqnarray}\label{BellAME}
    I_{\mathrm{AME}}&:=&\frac{1}{\sqrt{3}\,\lambda_1}\left[\langle(A_0+A_1+A_2)B_0D_0\rangle+\langle(A_0+\omega^2 A_1+\omega A_2)B_2C_0^2D_0\rangle\right]\nonumber\\
    &&+\frac{1}{2\sqrt{3}\,\lambda_1\omega}\left[\langle(A_0+\omega A_1+\omega^2A_2)B_1C_0D_0\rangle+\langle(A_0+\omega A_1+\omega^2 A_2)B_0C_0^2D_1\rangle\right]\nonumber\\
    &&+\langle B_0 C_1D_0\rangle+c.c.\leq\beta_{\mathrm{AME}}^C,
\end{eqnarray}
where c.c. stands for the complex conjugation of all five terms and represents
the expectation values obtained for the case $n=2$ of the assignment (\ref{assign1}); 
in particular, it makes the Bell expression real. Moreover, the second line comes with $1/2$ coefficient for reasons that will become clear later. The classical value in this case is 
\begin{equation}\label{AMECL}
 \beta_{\mathrm{AME}}^C=7.63816.
\end{equation}

Let us prove that the maximal quantum violation of this inequality 
is $\beta_{\mathrm{AME}}^Q=8$. First, denoting by $\mathcal{B}_{\mathrm{AME}}$ a Bell operator
constructed from $I_{\mathrm{AME}}$, we can write the following sum-of-squares decomposition
\begin{eqnarray}
    8\mathbbm{1}-\mathcal{B}_{\mathrm{AME}}&\!\!\!=\!\!\!&(\mathbbm{1}-\widetilde{A}_0B_0D_0)^{\dagger}
    (\mathbbm{1}-\widetilde{A}_0B_0D_0)+(\mathbbm{1}-\widetilde{A}_2B_2C_0^2D_0)^{\dagger}(\mathbbm{1}-\widetilde{A}_2B_2C_0^2D_0)\nonumber\\
    &&+\frac{1}{2}(\mathbbm{1}-\widetilde{A}_1B_1C_0D_0)^{\dagger}(\mathbbm{1}-\widetilde{A}_1B_1C_0D_0)+\frac{1}{2}(\mathbbm{1}-\widetilde{A}_1B_0C_0^2D_1)^{\dagger}(\mathbbm{1}-\widetilde{A}_1B_0C_0^2D_1)\nonumber\\
    &&+(\mathbbm{1}-B_0C_1D_0)^{\dagger}(\mathbbm{1}-B_0C_1D_0),
\end{eqnarray}
where $A_x$, $B_y$, etc. are arbitrary three-outcome unitary observables. To prove that this decomposition holds true one simply expands its right-hand side and uses the property [cf. Eq. (\ref{prop2})], which in the particular case $d=3$ reads,
\begin{equation}
    \widetilde{A}_0^{\dagger}\widetilde{A}_0+\widetilde{A}_1^{\dagger}\widetilde{A}_1+\widetilde{A}_2^{\dagger}\widetilde{A}_2=3\mathbbm{1}.
\end{equation}
Now it becomes clear why the second line of $I_{\mathrm{AME}}$ comes with $1/2$.

From this decomposition we immediately conclude that $8\mathbbm{1}-\mathcal{B}_{\mathrm{AME}}\geq 0$ for any choice of the local observables, which implies that also for any state $\ket{\psi}$, $\langle\psi|\mathcal{B}_{\mathrm{AME}}|\psi\rangle\leq 8$. 
To show that this bound is tight it suffices to provide a quantum realisation achieving it. Such a realisation can be constructed by inverting the transformation in Eqs. (\ref{assignd3}) and (\ref{assign5}), that is, by taking 
\begin{equation}
    A_x=\frac{\lambda_1}{\sqrt{3}}\sum_{k=0}^2\omega^{xk}\omega^{k(k+1)}O_{k}\qquad (i=0,1,2),
\end{equation}
and $B_y=ZX^y$ with $y=0,1,2$, $C_0=Z$ and $C_1=X$, and $D_w=ZX^w$ with $w=0,1$,
%
%
we can bring the Bell operator $\mathcal{B}_{\mathrm{AME}}$ to 
\begin{equation}
    \mathcal{B}_{\mathrm{AME}}=G_1+G_1G_2^2+\frac{1}{2}(G_1G_2+G_1G_4)+G_3+h.c.,
\end{equation}
which is simply a sum of the stabilizing operators of $|\mathrm{AME(4,3)}\rangle$.
As a result, the latter achieves the maximal quantum value of the Bell inequality 
(\ref{BellAME}).


\paragraph{Example 2: Two-qudit maximally entangled state.} Let us then consider the case of arbitrary prime $d$ and construct Bell inequalities for the simplest graph state which is the maximally entangled state (\ref{maxent}) stabilized by the two generators given in Eq. (\ref{maxentstab}).
Since we are now concerned with the bipartite scenario we can denote the observables measured by the parties by $A_x$ and $B_y$; the numbers of observables on both sites will be specified later. As already explained, to construct Bell inequalities we cannot simply use the replacement (\ref{replacement}), we rather need to employ the one in Eq. (\ref{assign1}). Let us moreover assume that we implement this transformation at Alice's site. 

To be able to apply the above assignments, we need to consider
a larger set of stabilizing operators which apart from $X$ and $Z^k$ 
operators contain also $(XZ^k)^n$ with $k=0,\ldots,d-1$ and $n=1,\ldots,d-1$. 
To construct such a set one can for instance take the following 
products of $G_i'$ given in Eq. \eqref{maxentstab}:
\begin{equation}
G_1'(G_2')^k=XZ^k\otimes XZ^{-k}    \qquad (k=0,1,\ldots,d-1).
\end{equation}
%
%
%
However, to take into account all the outcomes of the measurements performed by both parties we need to also include the powers of the above stabilizing operators [cf. Sec. \ref{BellSc}] which leads us to the following $d(d-1)$ stabilizing operators of $\ket{\psi_{d}^+}$:
\begin{equation}
\mathcal{G}^n_k:=\left[G_1'(G_2')^k\right]^n=\left(XZ^k\right)^n\otimes \left(XZ^{-k}\right)^n \qquad (k=0,\ldots,d-1;\,n=1,\ldots,d-1).
\end{equation}
We can now construct Bell inequalities maximally violated by the two-qudit maximally entangled states. Precisely, to each of the stabilizing operators $\mathcal{G}_k^n$ we associate an expectation value in which the particular matrices appearing at the first site are replaced by the combinations (\ref{assign1}) of the observables $A_x$, 
\begin{equation}
    \left(XZ^k\right)^n\,\to \, \widetilde{A}_k^{(n)},
\end{equation}
whereas at the second
site we substitute directly
\begin{equation}
    \left(XZ^{-k}\right)^{n}\,\to\,B_k^n.
\end{equation}
In other words, we associate
\begin{equation}
    \mathcal{G}_k^n\,\to\,\left\langle \widetilde{A}_k^{(n)}B_k^n\right\rangle
\end{equation}
with $k=0,1,2$ and $n=1,2$.
%
%
%

Adding then all the obtained expectation values and exploiting the fact that $\lambda_n^{-1}=\lambda_n^*$, we finally 
arrive at Bell inequalities derived previously in Ref. \cite{kaniewski2018maximal}: 
\begin{eqnarray}\label{Imax}
    I_{\max}&\!\!\!:=\!\!\!&\sum_{n=1}^{d-1}\sum_{k=0}^{d-1}\left\langle\widetilde{A}^{(n)}_kB_k^n\right\rangle
    \nonumber\\
  &\!\!\!=\!\!\!&  \frac{1}{\sqrt{d}}\sum_{n=1}^{d-1}\lambda_n^*
    \sum_{x,y=0}^{d-1}\omega^{-nxy}\left\langle A_x^n B_y^n\right\rangle\leq \beta_{\max}^{C},
\end{eqnarray}
where $\beta_{\max}^{C}$ stands for the maximal classical value of $I_{\max}$. It is in general difficult to compute $\beta_{\max}^{C}$ analytically, however, for the lowest values of $d=3,5,7$ it was found numerically in \cite{kaniewski2018maximal}; for completeness we listed these values in Table \ref{tab:Table}. 
\begin{table}[]
    \centering
    \begin{tabular}{|c|c|c|c|}
    \hline
         $d$ & $\beta_L$ & $\beta_Q$& $\beta_Q/\beta_L$ \\
         \hline 
         3   & $6\cos(\pi/9)$  & 6 & 1.064\\[1ex]
         5   & $4(2+\sqrt{5})$ & 20 & 1.1803\\[1ex]
         7   &    $\simeq 33.3494$             & 42 &1.2594\\
         \hline
    \end{tabular}
    \caption{Maximal classical values of the Bell expression $I_{\max}$ given in Eq. (\ref{Imax}) for $d=3,5,7$. For comparison we also present the maximal quantum values.}
    \label{tab:Table}
\end{table}
%

On the other hand, these Bell inequalities are designed so that their maximal quantum value can be determined straightforwardly. Let us formulate and prove the following fact.
\begin{fakt}
The maximal quantum value of the Bell expressions $I_{\max}^{(d)}$ is $\beta_{\max}^{Q}=d(d-1)$.
\end{fakt}
\begin{proof}
The proof is straightforward and consists of two steps. First, we denote by 
\begin{equation}
    \mathcal{B}_{\max}=\frac{1}{\sqrt{d}}\sum_{n=1}^{d-1}\frac{1}{\lambda_n}
    \sum_{x,y=0}^{d-1}\omega^{-nxy} A_x^n \otimes B_y^n
\end{equation}
a Bell operator associated to the expression $I_{\max}^{(d)}$, where $A_x$ and $B_y$ are arbitrary $d$-outcome unitary observables. Second, one uses Eq. (\ref{prop2}) as well as the fact that the Bell operator is Hermitian to observe that the following sum-of-squares decomposition holds true
\begin{equation}
d(d-1)\mathbbm{1}-\mathcal{B}_{\max}=\frac{1}{2}\sum_{n=1}^{d-1}\sum_{y=0}^{d-1}
\left(\mathbbm{1}-\widetilde{A}_y^{(n)}\otimes B_y^n\right)^{\dagger}\left(\mathbbm{1}-\widetilde{A}_y^{(n)}\otimes B_y^n\right).
\end{equation}
Consequently, $d(d-1)\mathbbm{1}-\mathcal{B}_{\max}$ is a positive semi-definite operator
for any choice of local observables, and thus $\beta_{\max}^{(d)}\leq  d(d-1)$. 
To prove that this inequality is tight we can construct a quantum realisation for which $I_{\max}^{(d)}=d(d-1)$. Precisely, we notice that for the following choice of observables for Alice and Bob [cf. Eq. (\ref{AObs})],
\begin{equation}
    A_x^{n}=\frac{\lambda_n}{\sqrt{d}}\sum_{k=0}^{d-1}\omega^{nxk}\omega^{nk(k+1)}(XZ^k)^n,\qquad
    B_y^n=(XZ^{-k})^{n}
\end{equation}
the Bell operator $\mathcal{B}_{\max}$ simply becomes a sum of
the stabilizing operators of $\ket{\psi_{d}^+}$,
\begin{equation}
    \mathcal{B}_{\max}=\sum_{n=1}^{d-1}\sum_{k=0}^{d-1} \left[G'_1(G_2')^k\right]^n,
\end{equation}
meaning that $\langle\psi_{d}^+|\mathcal{B}_{\max}|\psi_d^+\rangle=d(d-1)$. As a result
$\beta_{\max}^{(d)}=d(d-1)$, which completes the proof.
\end{proof}

\subsection{General construction}\label{sec: general construction}

We are now ready to provide our general construction of Bell inequalities for
arbitrary graph states. Let us first set the notation. 

Consider a graph $\mathcal{G}=(\mathcal{V},\mathcal{E},\mathcal{R},d)$ and choose two of its vertices that are connected. Without any loss of generality we can label them by $1$ and $2$. Let then $\N_1$ and $N_1$ be respectively the neighbourhood of the first vertex, i.e., the set of all vertices that are connected to it, and its cardinality. Clearly, we can relabel all the other neighbours of vertex $1$ by $j \in \N_1 \setminus \{2\} \equiv \{ 3,\ldots,N_1+1 \}$. We finally label the remaining vertices that are not connected to the first vertex as $l \in \mathcal{V} \setminus \{1,\N_1 \} \equiv \{ N_1+2,\ldots, N \}$. The generators corresponding to the graph $\mathcal{G}$ are denoted $G_i$ [see Eq. (\ref{Gi-m}) for the definition thereof], whereas 
the graph state stabilized by them by $\ket{G}$.
%
%

Let us then define the Bell scenario. 
It will be beneficial for our construction to slightly modify the way we denote the 
observers and the observables they measure. Precisley, the observables measured by the first two parties are denoted by $A_x$ and $B_y$ with $x,y=0,\ldots,d-1$, respectively; notice that both them can choose among $d$ different settings. Then, the other observers connected to the first party $A$ measure three observables which we denote 
$C^{(i)}_z$ with $z=0,1,2$ and $i\in\mathcal{N}_1\setminus\{2\}$. The remaining observers (that do not belong to $\mathcal{N}_1$) have only two observables at their disposal, denoted $D^{(i)}_{0},D^{(i)}_{1}$ where $i \in  \{ N_1+2,\ldots, N \}$.

To derive a Bell inequality tailored to the graph state $\ket{G}$  we begin by rewriting the stabilizing operators $G_i$ corresponding to $\mathcal{G}$ by explicitly presenting operators acting on the first two sites as well as on the neighbourhood $\mathcal{N}_1$. The first two stabilizing operators read
\begin{equation}\label{StabOp2}
    G_1=X_1\otimes Z^{r_{1,2}}_2\otimes\bigotimes_{m\in \mathcal{N}_1\setminus\{2\}}Z_m^{r_{1,m}}
\end{equation}
and
\begin{equation}
    G_2=Z^{r_{1,2}}_1\otimes X_2\otimes\bigotimes_{m\in \mathcal{N}_1\setminus\{2\}}Z_m^{r_{2,m}}\otimes\bigotimes_{m=N_1+2}^N Z_m^{r_{2,m}}.
\end{equation}
Then, those associated to the other vertices belonging to $\mathcal{N}_1$ are given by
%
%
\begin{equation}
    G_i=Z^{r_{1,i}}_1\otimes Z_2^{r_{2,i}}\otimes X_i\otimes\bigotimes_{m\in \mathcal{N}_1\setminus\{2,i\}}Z_m^{r_{i,m}}\otimes\bigotimes_{m=N_1+2}^NZ_m^{r_{i,m}},
\end{equation}
where $j=3,\ldots,N_1$, whereas the remaining 
$G_i$'s for $i\in\{N_1+2,\ldots,N\}$ are of the following form
\begin{equation}
    G_i=\mathbbm{1}_1\otimes Z_2^{r_{2,i}}\otimes \bigotimes_{m\in \mathcal{N}_1\setminus\{2\}}Z_m^{r_{i,m}}\otimes X_i\otimes\bigotimes_{m\in \{N_1+2,\ldots,N\}\setminus\{i\}}Z_m^{r_{i,m}},
\end{equation}
It is worth adding here that since by assumption the first two vertices are connected, 
$r_{1,2}\neq 0$. Moreover, $G_1$ acts trivially on all sites that are outside $\mathcal{N}_1\cup \{1\}$.

Given the stabilizing operators, let us then follow the procedure outline already 
in the previous examples. We begin by constructing a suitable set of 
stabilizing operators. First, to create at the first site the operators
$XZ^k$ required for the assignment (\ref{assign1}), we consider
products $G_1G_2^k$ with $k=0,\ldots,d-1$. This set, however, does not
uniquely define the graph state $\ket{G}$ as it lacks
the other generators. To include them we first notice that any
$G_i$ with $i\in\mathcal{N}_1\setminus\{2\}$ contains the 
$Z$ operator or its power at the first position and therefore
we take their products with $G_1$, that is, 
$G_1G_i$ with $i\in\mathcal{N}_1\setminus\{2\}$, again to obtain $XZ^k$ at the first site. 
On the other hand, the remaining generators $G_i$ for $i\in\{N_1+2,\ldots,N\}$
have the identity at the first position and therefore 
we directly add them to the set.

Thus, the total list of the stabilizing operators that we use to construct a Bell inequality is
\begin{equation}\label{StabOp2}
\begin{array}{ccll}
    \mathcal{G}_{1,k}^n&:=&(G_1G_2^k)^n\qquad &(k=0,\ldots,d-1),\\[1ex]
    \mathcal{G}_{2,k}^n&:=&\left(G_1G_k\right)^n\qquad &(k=3,\ldots,N_1+1),\\[1ex]
    \mathcal{G}_{3,k}^n&:=&G_k^n&(k=N_1+2,\ldots,N),
\end{array}
\end{equation}
where we have added powers to include all outcomes in the Bell scenario.
Let us now write these operators explicitly
\begin{equation}\label{G1}
    \mathcal{G}^n_{1,k}=\left(XZ^{kr_{1,2}}\right)_1^n\otimes \left(Z^{r_{1,2}}X^k\right)_2^n\otimes\bigotimes_{m\in\mathcal{N}_1\setminus\{2\}}Z_m^{n(r_{1,m}+kr_{2,m})}
    \otimes\bigotimes_{m\in\{N_1+2,\ldots,N\}}Z_m^{nk r_{2,m}}
\end{equation}
for $k=0,\ldots,d-1$,
\begin{equation}\label{G2}
    \mathcal{G}^n_{2,k}=\left(XZ^{r_{1,k}}\right)^n_1\otimes Z^{n(r_{1,2}+r_{2,k})}_2\otimes (Z^{r_{1,k}}X)_k^n\otimes\bigotimes_{m\in\mathcal{N}_1\setminus\{2,k\}}Z_m^{n(r_{1,m}+r_{k,m})}
    \otimes\bigotimes_{m\in\{N_1+2,\ldots,N\}}Z_m^{nr_{k,m}}
\end{equation}
for $k=3,\ldots,N_1+2$, and
\begin{equation}\label{G3}
    \mathcal{G}^n_{3,k}=\mathbbm{1}_1\otimes Z_2^{nr_{2,k}}\otimes \bigotimes_{m\in \mathcal{N}_1\setminus\{2\}}Z_m^{nr_{k,m}}\otimes X^n_k\otimes\bigotimes_{m\in \{N_1+2,\ldots,N\}\setminus\{k\}}Z_m^{nr_{k,m}}
\end{equation}
for $k\in\{N_1+2,\ldots,N\}$.

We associate to each of these stabilizing operators an expectation value
in which the local operators are replaced by $d$-outcome observables
or combinations thereof. Let us begin with the first site where
we have $(XZ^{kr_{1,2}})^n$ with $k=0,\ldots,d-1$, $XZ^{r_{1,i}}$ with 
$i=3,\ldots,N_1$ and the identity. It is important to notice here
that due to the fact that $d$ is a prime number, for any $r_{1,2}\neq 0$, $kr_{1,2}$ 
spans the whole set $\{0,\ldots,d-1\}$ for $k=0,\ldots,d-1$; in other words, 
the function $f(k)=kr_{1,2}$ defined on the set $\{0,\ldots,d-1\}$ is a one-to-one function. 
Thus, $XZ^{kr_{1,2}}$ contains all the $d$ different matrices appearing in the transformation (\ref{assign1}).
We thus substitute
\begin{equation}
(XZ^{kr_{1,2}})^n\,\to\,\widetilde{A}_{kr_{1,2}}^{(n)}:= \frac{\omega^{-nkr_{1,2}(kr_{1,2}+1)}}{\sqrt{d}\,\lambda_n}\sum_{x=0}^{d-1}\omega^{-nkr_{1,2}x}A_x^n.    
\end{equation}
Analogously, we substitute
\begin{equation}
(XZ^{r_{1,i}})^n\,\to\,\widetilde{A}_{r_{1,i}}^{(n)}:= \frac{\omega^{-nr_{1,i}(kr_{1,i}+1)}}{\sqrt{d}\,\lambda_n}\sum_{x=0}^{d-1}\omega^{-nr_{1,i}x}A_x^n    \end{equation}
for $i=3,\ldots,N_1+1$; in both cases $n=1,\ldots,d-1$. 

Let us then move to the second site. The matrices appearing there 
are $Z^{r_{1,2}}X^k$ with $k=0,\ldots,d-1$
and $Z^{n(r_{1,2}+r_{2,i})}$ with $i=3,\ldots,N_1+1$. Since for any $r_{1,2}$
the former are all proper observables in our scenario, that is, they are unitary and
their spectra belong to $\{1,\omega^1,\ldots,\omega^{d-1}\}$, we can directly 
substitute them by observables $B_k$. Specifically, for $k=0$ we assign
\begin{equation}
Z^n\,\to\, B_0^n
\end{equation}
which implies in particular that
\begin{equation}
    Z^{nr_{1,2}}\,\to\, B_0^{nr_{1,2}},
\end{equation}
and for the remaining $k=1,\ldots,d-1$,
\begin{equation}
    (Z^{r_{1,2}}X^k)^n\,\to\, B_{k}^n.
\end{equation}
We distinguish the case $k=0$ to simplify the assignment of observables to 
the other set of matrices $Z^{n(r_{1,2}+r_{2,i})}$ with $i=3,\ldots,N_1+1$. 
These are simply powers of $Z$ and thus we associate with them a single observable $B_0$; precisely,
\begin{equation}
    Z^{n(r_{1,2}+r_{2,i})}\to B_0^{n(r_{1,2}+r_{2,i})}.
\end{equation}

Let us now consider all sites from $\mathcal{N}_1\setminus\{2\}$.
From Eqs. (\ref{G1}), (\ref{G2}) and (\ref{G3}) 
it follows that 
the operators appearing there are $Z^{r_{1,i}}X$ with $i=3,\ldots,N_1+1$
and powers of $Z$, and thus we can make the following replacements
\begin{equation}
Z\,\to\, C_0^{(i)}\qquad\mathrm{and}\qquad
    Z^{r_{1,i}}X\,\to\, C_1^{(i)}
\end{equation}
for any $i=3,\ldots,N_1$. Finally, for the remaining sites we have
simply the $X$ operator at various sites and powers of $Z$. Thus,
for any $i=N_1+2,\ldots,N$,
\begin{equation}
Z\,\to\, D_0^{(i)}\qquad\mathrm{and}\qquad
    X\,\to\, D_1^{(i)}.
\end{equation}

Collecting all these substitutions together we have
\begin{equation}\label{G02}
\mathcal{G}_{1,0}^n\,\to\,\widetilde{\mathcal{G}}_{1,0}^{(n)}:= \widetilde{A}_{0}^{(n)}\otimes B_0^{nr_{1,2}}\otimes\bigotimes_{i=3}^{N_1+1}\left[C^{(i)}_0\right]^{nr_{1,i}}
\end{equation}
and
\begin{equation}\label{G12}
    \mathcal{G}_{1,k}^n\,\to\,\widetilde{\mathcal{G}}_{1,k}^{(n)}:= \widetilde{A}_{kr_{1,2}}^{(n)}\otimes B_k^{n}\otimes\bigotimes_{i=3}^{N_1+1} \left[C^{(i)}_0\right]^{n(r_{1,i}+kr_{2,i})}\otimes\bigotimes_{i=N_1+2}^{N}\left[D^{(i)}_0\right]^{nkr_{2,i}}
\end{equation}
for $k=1,\ldots,d-1$. Then, 
\begin{equation}\label{G22}
    \mathcal{G}_{2,k}^{n}\,\to\,\widetilde{\mathcal{G}}_{2,k}^{(n)}:= \widetilde{A}_{r_{1,k}}^{(n)}\otimes B_0^{n(r_{1,k}+r_{2,k})}\bigotimes_{i=3}^{k-1}\left[C_0^{(i)}\right]^{n(r_{1,i}+r_{k,i})}\otimes \left[C^{(k)}_1\right]^n\otimes\bigotimes_{i=k+1}^{N_1+1}\left[C_0^{(i)}\right]^{n(r_{1,i}+r_{k,i})}\bigotimes_{i=N_1+2}^{N}\left[D_0^{(i)}\right]^{nr_{k,i}}
\end{equation}
with $k\in\{3,\ldots,N_1+1\}$, and, finally,
\begin{equation}\label{G32}
    \mathcal{G}_{3,k}^{n}\,\to\,\widetilde{\mathcal{G}}_{3,k}^{(n)}:= B_0^{nr_{2,k}}\bigotimes_{i=3}^{N_1+1}\left[C_0^{(i)}\right]^{nr_{k,i}}\bigotimes_{i=N_1+2}^{k-1}\left[D_0^{(i)}\right]^{nr_{k,i}}\otimes\left[ D_1^{(k)}\right]^n\otimes\bigotimes_{i=k+1}^{N}\left[D_0^{(i)}\right]^{nr_{k,i}}
\end{equation}
for $k\in\{N_1+2,\ldots,N\}$.

Lastly, by taking a weighted sum of expectation values of the above operators, we arrive at the following class of Bell expressions for a given graph state:
\begin{eqnarray}\label{genBell2}
    I_{\mathcal{G}}:=\sum_{n=1}^{d-1}
    \left[\left\langle\widetilde{\mathcal{G}}_{1,0}^{(n)}\right\rangle
 +\sum_{k=1}^{d-1}c_{1,k}\left\langle\widetilde{\mathcal{G}}_{1,k}^{(n)}\right\rangle
+\sum_{k=3}^{N_1+1}c_{2,k}\left\langle\widetilde{\mathcal{G}}_{2,k}^{(n)}\right\rangle+
    \sum_{k=N_1+2}^N\left\langle\widetilde{\mathcal{G}}_{3,k}^{(n)}\right\rangle\right],
\end{eqnarray}
%
%
%
where $c_{i,k}>0$ are some free parameters that satisfy
\begin{equation}\label{condts1}
    c_{1,k}+\sum_{\substack{j=3\\\{j:r_{1,j}=kr_{1,2}\}}}^{N_1+1}c_{2,j}=1
\end{equation}
for each $k=1,\ldots,d-1$, where the second sum goes over all $j$ such that for a fixed $k$, $r_{1,j}=kr_{1,2}$. As we will see below the conditions (\ref{condts1}) are used for constructing sum-of-squares decompositions of the Bell operators corresponding to 
$I_{\mathcal{G}}$, which in turn are crucial for determining the maximal quantum values of $I_{\mathcal{G}}$. In fact, we can prove the following theorem.

\begin{thm}The maximal quantum value of $I_{\mathcal{G}}$ is 
\begin{equation}\label{MaxQVal}
    \beta_{\mathcal{G}}^Q=(d-1)(N-N_1+d-1).
\end{equation}
\end{thm}
\begin{proof}
To prove this statement let us consider a Bell operator corresponding to $I_{\mathcal{G}}$,
\begin{equation}
    \mathcal{B}_{\mathcal{G}}=\sum_{n=1}^{d-1}
    \left[\widetilde{\mathcal{G}}_{1,0}^{(n)}
 +\sum_{k=1}^{d-1}c_{1,k}\,\widetilde{\mathcal{G}}_{1,k}^{(n)}
+\sum_{k=3}^{N_1+1}c_{2,k}\,\widetilde{\mathcal{G}}_{2,k}^{(n)}+
    \sum_{k=N_1+2}^N\widetilde{\mathcal{G}}_{3,k}^{(n)}\right],
\end{equation}
where $\widetilde{\mathcal{G}}_{i,k}^{(n)}$ are defined in Eqs. 
(\ref{G02})-(\ref{G32}). We
show that $\mathcal{B}_{\mathcal{G}}$ 
admits the following sum-of-squares decomposition
\begin{eqnarray}\label{deco24}
    \mathcal{B}_{\mathcal{G}}&\!\!\!=\!\!\!&(d-1)(N-N_1+d-1)\mathbbm{1}\nonumber\\
    &&-\frac{1}{2}\sum_{n=1}^{d-1}
    \left[
   \left(\mathbbm{1}-\widetilde{\mathcal{G}}^{(n)}_{1,0}\right)^{\dagger}\left(\mathbbm{1}-\widetilde{\mathcal{G}}^{(n)}_{1,0}\right)+\sum_{k=1}^{d-1}c_{1,k}\left(\mathbbm{1}-\widetilde{\mathcal{G}}^{(n)}_{1,k}\right)^{\dagger}\left(\mathbbm{1}-\widetilde{\mathcal{G}}^{(n)}_{1,k}\right)
    \right.\nonumber\\
    &&\hspace{1.5cm}\left.+\sum_{k=3}^{N_1+1}c_{2,k}\left(\mathbbm{1}-\widetilde{\mathcal{G}}^{(n)}_{2,k}\right)^{\dagger}\left(\mathbbm{1}-\widetilde{\mathcal{G}}^{(n)}_{2,k}\right)+\sum_{k=N_1+2}^{N}\left(\mathbbm{1}-\widetilde{\mathcal{G}}^{(n)}_{3,k}\right)^{\dagger}\left(\mathbbm{1}-\widetilde{\mathcal{G}}^{(n)}_{3,k}\right)\right].\nonumber\\
\end{eqnarray}

To verify that this decomposition holds true let us expand the expression appearing in the square brackets for a particular $n$,
\begin{eqnarray}\label{rownanie00}
   &&\hspace{-1cm}\left(1+\sum_{k=1}^{d-1}c_{1,k}+\sum_{k=3}^{N_1+1}c_{2,k}+N-N_1-1\right)\mathbbm{1}-\mathcal{B}_{\mathcal{G}}^{(n)}-\left[\mathcal{B}_{\mathcal{G}}^{(n)}\right]^{\dagger}
   +\left(\widetilde{\mathcal{G}}^{(n)}_{1,0}\right)^{\dagger}
    \widetilde{\mathcal{G}}^{(n)}_{1,0}\nonumber\\
    &&+\sum_{k=1}^{d-1}c_{1,k}\left(\widetilde{\mathcal{G}}^{(n)}_{1,k}\right)^{\dagger}\widetilde{\mathcal{G}}^{(n)}_{1,k}
    +\sum_{k=3}^{N_1+1}c_{2,k}\left(\widetilde{\mathcal{G}}^{(n)}_{2,k}\right)^{\dagger}\widetilde{\mathcal{G}}^{(n)}_{2,k}+\sum_{k=N_1+2}^{N}\left(\widetilde{\mathcal{G}}^{(n)}_{3,k}\right)^{\dagger}\widetilde{\mathcal{G}}^{(n)}_{3,k},
   \end{eqnarray}
where $\mathcal{B}^{(n)}_{\mathcal{G}}$ is a part of the Bell operator corresponding to 
a particular $n$, that is, 
\begin{equation}
\mathcal{B}_{\mathcal{G}}^{(n)}=\widetilde{\mathcal{G}}_{1,0}^{(n)}
+\sum_{k=1}^{d-1}c_{1,k}\,
\widetilde{\mathcal{G}}_{1,k}^{(n)}+\sum_{k=3}^{N_1+1}c_{2,k}\,
\widetilde{\mathcal{G}}_{2,k}^{(n)}+\sum_{k=N_1+2}^{N}c_{3,k}\,
\widetilde{\mathcal{G}}_{3,k}^{(n)}.
\end{equation}
We now notice that by summing all the conditions (\ref{condts1}) one can deduce that
\begin{equation}\label{condt2}
    \sum_{k=1}^{d-1}c_{1,k}+\sum_{k=3}^{N_1+1}c_{2,k}=d-1,
\end{equation}
which implies that the coefficient in front of
the identity simplifies to $d+N-N_1-1$. Using the definitions of $\widetilde{\mathcal{G}}_{i,k}^{(n)}$ one then has that
\begin{eqnarray}
    &&\left(\widetilde{\mathcal{G}}^{(n)}_{1,0}\right)^{\dagger}
    \widetilde{\mathcal{G}}^{(n)}_{1,0}+\sum_{k=1}^{d-1}c_{1,k}\left(\widetilde{\mathcal{G}}^{(n)}_{1,k}\right)^{\dagger}\widetilde{\mathcal{G}}^{(n)}_{1,k}
    +\sum_{k=3}^{N_1+1}c_{2,k}\left(\widetilde{\mathcal{G}}^{(n)}_{2,k}\right)^{\dagger}\widetilde{\mathcal{G}}^{(n)}_{2,k}+\sum_{k=N_{1}+2}^N\left(\widetilde{\mathcal{G}}^{(n)}_{3,k}\right)^{\dagger}\widetilde{\mathcal{G}}^{(n)}_{3,k}\nonumber\\
    &&\hspace{1cm}=\left(\widetilde{A}_0^{(n)}\right)^{\dagger}
    \widetilde{A}_0^{(n)}+\sum_{k=1}^{d-1}c_{1,k}\left(\widetilde{A}_{kr_{1,2}}^{(n)}\right)^{\dagger}
    \widetilde{A}_{kr_{1,2}}^{(n)}+\sum_{k=3}^{N_1+1}c_{2,k}\left(\widetilde{A}_{r_{1,k}}^{(n)}\right)^{\dagger}
    \widetilde{A}_{r_{1,k}}^{(n)}+(N-N_1-1)\mathbbm{1}\nonumber\\
    &&\hspace{1cm}=\sum_{k=0}^{d-1}\left(\widetilde{A}_k^{(n)}\right)^{\dagger}
    \widetilde{A}_k^{(n)}+(N-N_1-1)\mathbbm{1}=(d+N-N_1-1)\mathbbm{1},
\end{eqnarray}
where the second line follows from the fact that apart from the first position all the local operators in
$\widetilde{\mathcal{G}}_{i,k}^{(n)}$ are unitary (notice also that $\widetilde{\mathcal{G}}_{3,k}^{(n)}$ have the identity at the first position), whereas the second line stems from the conditions (\ref{prop2}) and (\ref{condts1}).
All this allows us to rewrite (\ref{rownanie00}) simply as $2(d+N-N_1-1)\mathbbm{1}-\mathcal{B}_{\mathcal{G}}^{(n)}-\mathcal{B}_{\mathcal{G}}^{(n)\dagger}$. Taking finally the sum of these terms over $n=1,\ldots,d-1$ we arrive at the decomposition (\ref{deco24}), which completes the first part of the proof.

From the decomposition (\ref{deco24}) one directly infers that 
$(d-1)(d+N-N_1-1)\mathbbm{1}-\mathcal{B}_{\mathcal{G}}$ is a positive semi-definite
operator for any choice of the local observables, which is equivalent to 
say that for any Bell operator $\mathcal{B}_{\mathcal{G}}$ corresponding to 
$I_{\mathcal{G}}$ and any pure state $\ket{\psi}$, the following inequality 
is satisfied
\begin{equation}
    \langle\psi|\mathcal{B}_{\mathcal{G}}|\psi\rangle\leq (d-1)(d+N-N_1-1).
\end{equation}
To show that this inequality is tight, and at the same time complete the proof, let us provide a particular quantum realisation that achieves it. To this end, we can invert the transformation we used to construct $I_{\mathcal{G}}$. Precisely, we let the first party measure $d$ observables $A_{k}$ with $k=0,\ldots,d-1$ which are defined in Eq. (\ref{AObs}); for them $\widetilde{A}_k^{(n)}=(XZ^k)^n$. The remaining parties measure
\begin{equation}
    B_0^n=Z^n,\qquad B_k^n=(Z^{r_{1,2}}X^k)^n\qquad (k=0,\ldots,d-1)
\end{equation}
\begin{equation}
    C_0^{(i)}=Z,\qquad C_1^{(i)}=Z^{r_{1,i}}X
\end{equation}
for $i=3,\ldots,N_1+1$, and, finally,
\begin{equation}
    D_0^{(i)}=Z,\qquad D_1^{(i)}=X
\end{equation}
for $i=N_1+2,\ldots,N$.

It is not difficult to see that for this choice of quantum observables
the Bell operator reduces to a combination of the stabilizing operators of
the given graph state $\ket{G}$, that is, 
\begin{equation}
    \mathcal{B}_{\mathcal{G}}=\sum_{n=1}^{d-1}\left[G_1^n+\sum_{k=1}^{d-1}c_{1,k}
    (G_1G_2^k)^n+\sum_{k=3}^{N_1+1}c_{2,k}(G_1G_k)^n+\sum_{k=N_1+2}^{N}G_k^n\right].
\end{equation}
Owing to the conditions (\ref{condts1}) as well as (\ref{condt2}), one finds that 
\begin{equation}
    \langle G|\mathcal{B}_{\mathcal{G}}|G\rangle=(d-1)(N-N_1+d-1),
\end{equation}
which is what we aimed to prove.
\end{proof}

We have thus obtained a family of Bell expressions whose maximal quantum 
values are achieved by graph states of arbitrary prime local dimension.
To turn them into nontrivial Bell inequalities one still needs to 
determine their maximal classical values which is in general a hard task. 
For the simplest cases such as Bell inequalities for the AME(4,3) state or those 
tailored to the maximally entangled state of two qudits for low $d$'s, the classical bounds can be determined numerically [cf. Eq. \eqref{AMECL} and Table \ref{tab:Table}]. On the other hand, in the next section we show that our inequalities allow to self-test the graph states of local dimension three, and thus for all of them the classical bound is strictly lower than the Tsirelson's bound. It is also worth mentioning that the ratio between the maximal quantum and classial values will certainly depend on the choice of vertices 1 and 2, in particular on the number of neighbours of the first vertex $N_1$ because this number appears in the formula for $\beta_Q$ (\ref{MaxQVal}).

Let us finally mention that our inequalities are scalable in the sense that the number of expectation values they are constructed from scales linearly with $N$. Indeed, it follows from Eq. (\ref{genBell2}) that the number of expectation values in $I_{\mathcal{G}}$ is 
\begin{equation}
    (d-1)[N+(N_1+d)(d-1)]
\end{equation}
which in the worst case $N_1=N-1$ reduces to $(d-1)[Nd+(d-1)^2]$. This number can still be lowered twice because the expectation values in $I_{\mathcal{G}}$ for $n=\lceil d/2\rceil,\ldots,d-1$ are complex conjugations of those for $n=1,\ldots,\lfloor d/2\rfloor$. Another possibility for lowering it number is to choose as the first vertex the one with the lowest neighbourhood. While it is an interesting question whether it is possible to design another construction which requires measuring even less expectation values, it seems that the linear scaling in $N$ is the best one can hope for.

\section{Self-testing of qutrit graph states}
\label{Sec:Self-testing}

Here we show our second main result: we demonstrate that our Bell inequalities can be used to self-test arbritrary graph states of local dimension $d = 3$. In this particular case the general Bell expression (\ref{genBell2}) can be written as
\begin{eqnarray}\label{genBell22}
    I_{\mathcal{G}}:=
    \left\langle\widetilde{\mathcal{G}}_{1,0}^{(n)}\right\rangle
 +\sum_{k=1}^{d-1}c_{1,k}\left\langle\widetilde{\mathcal{G}}_{1,k}^{(n)}\right\rangle
+\sum_{k=3}^{N_1+1}c_{2,k}\left\langle\widetilde{\mathcal{G}}_{2,k}^{(n)}\right\rangle+
    \sum_{k=N_1+2}^N\left\langle\widetilde{\mathcal{G}}_{3,k}^{(n)}\right\rangle+c.c.,
\end{eqnarray}
or explicitly as,
\begin{eqnarray}
    I_{\mathcal{G}}&\!\!\!:=\!\!\!&\left\langle \widetilde{A}_{0}B_0^{r_{1,2}}\prod_{i=3}^{N_1+1}\left[C^{(i)}_0\right]^{r_{1,i}}\right\rangle \nonumber\\
    &&
    +\sum_{k=1}^{2}c_{1,k}\left\langle \widetilde{A}_{kr_{1,2}}B_k\prod_{i=3}^{N_1+1}\left[C^{(i)}_0\right]^{r_{1,i}+kr_{2,i}}\prod_{i=N_1+2}^{N}\left[D^{(i)}_0\right]^{kr_{2,i}}\right\rangle \nonumber\\
    &&+\sum_{k=3}^{N_1+1}c_{2,k}\left\langle \widetilde{A}_{r_{1,k}}B_0^{r_{1,k}+r_{2,k}}\prod_{i=3}^{k-1}\left[C^{(i)}_0\right]^{r_{1,i}+r_{k,i}}\,C^{(k)}_1\,\prod_{i=k+1}^{N_1+1}\left[C^{(i)}_0\right]^{r_{1,i}+r_{k,i}}\prod_{i=N_1+2}^{N}\left[D_0^{(i)}\right]^{r_{k,i}}\right\rangle\nonumber\\
    &&+\sum_{k=N_1+2}^N\left\langle B_0^{r_{2,k}}\prod_{i=3}^{N_1+1}\left[C^{(i)}_0\right]^{r_{k,i}}\prod_{i=N_1+2}^{k-1}\left[D_0^{(i)}\right]^{r_{k,i}}\,D_1^{(k)}\,\prod_{i=k+1}^{N}\left[D_0^{(i)}\right]^{r_{k,i}}\right\rangle+c.c.,
\end{eqnarray}
where $c.c.$ stands for the complex conjugation and represents the $n=2$ term in Eq. (\ref{genBell2}), whereas the coefficients $c_{1,k}$ and $c_{2,k}$ satisfy the condition (\ref{condts1}). 

Let us now prove that maximal violation of $I_{\mathcal{G}}$ can be used to self-test the corresponding graph state according to Definition \ref{definition}. To this aim, we state the following theorem.
\begin{thm}
    Consider a connected graph $G$ and assume that the maximal quantum value of the corresponding Bell expression $I_{\mathcal{G}}$ is achieved by a pure state $\ket{\psi}\in\mathcal{H}_1\otimes\ldots\otimes\mathcal{H}_N$ and observables $A_x$, $B_y$, etc. acting on the local Hilbert spaces $\mathcal{H}_i$. Then, each Hilbert space $\mathcal{H}_i$ decomposes as $\mathcal{H}_i=\mathbbm{C}^3\otimes\mathcal{H}_i'$ and there exist local unitary operators $U_i$ with $i=1,\ldots,N$ such that
    \begin{equation}
     (U_1\otimes\ldots\otimes U_N)\ket{\psi}=\ket{\psi_{G}}\otimes\ket{\mathrm{aux}}
    \end{equation}
    with $\ket{\mathrm{aux}}$ being some state from the auxiliary Hilbert space $\mathcal{H}_1'\otimes\ldots\otimes\mathcal{H}_N'$.
\end{thm}

Before we present our proof let us mention that it is follows a similar reasoning to the proof of self-testing of $N$-qubit graph states in \cite{baccari2018scalable}, but since we deal here with qutrits it also makes a use of one of the results of Ref. \cite{kaniewski2018maximal}, which for completeness we state in Appendix \ref{Appendix} as Fact \ref{dupablada}. 

\begin{proof}Let us first notice that it is convenient to assume that the local reduced density matrices of the state $\ket{\psi}$ are full rank; otherwise we are able to characterize the observables only on the supports of these reduced density matrices. Moreover, we assume for simplicity that $r_{1,2}=1$; recall that by construction $r_{1,2}\neq 0$. The proof for the other case of $r_{1,2}=2$ goes along the same lines.

The sum-of-squares decomposition (\ref{deco24}) implies the following relations for the state and observables that achieve the maximal quantum value of the Bell expression $I_{\mathcal{G}}$,
\begin{equation}\label{d3cond1}
    \widetilde{\mathcal{G}}_{1,k}^{(n)}\ket{\psi}=\ket{\psi}
\end{equation}
for $k=0,1,2$, 
\begin{equation}\label{d3cond2}
      \widetilde{\mathcal{G}}_{2,k}^{(n)}\ket{\psi}=\ket{\psi}
\end{equation}
for $k=3,\ldots,N_1+1$, and
\begin{equation}\label{d3cond3}
      \widetilde{\mathcal{G}}_{3,k}^{(n)}\ket{\psi}=\ket{\psi}
\end{equation}
for $k=N_1+2,\ldots,N$. 

Before we employ the above relations in order to prove our self-testing statement let us recall that $\widetilde{A}_x^{(n)}$ $(x=0,1,2)$ are combinations of the first party's observables and are not unitary in general; still, they satisfy $\widetilde{A}_x^{(2)}=\widetilde{A}_{x}^{(1)\dagger}$. 
At the same time $B_y$, $C_z^{(i)}$, and $D_{w}^{(i)}$ are all unitary observables
which in the particular case $d=3$ satisfy $B_y^2=B_y^{\dagger}$ etc.
This implies that $\widetilde{\mathcal{G}}_{i,k}^{(2)}=\widetilde{\mathcal{G}}_{i,k}^{(1)\dagger}$.

The main technical step we need is to identify at each site two unitary observables whose anticommutator is unitary. This allows us to make use of Fact \ref{dupablada}
and Corollary \ref{Corollary} (see Appendix \ref{Appendix}) to define local unitary operators that map the two unkown observables to the qutrit ones. For parties having three measurement choices, the remaining observable will be directly mapped to other qutrit operators thanks to anticommutation relations that can be inferred from the sum-of-squares decompositions.

Our proof is quite technical and long and therefore to make it easier to follow we divide it into a few steps. In the first four we characterize every party's observables that give rise to the maximal quantum violation of the inequality, while in the last one we
prove the self-testing statement for the state.\\

\noindent \textbf{Step 1. ($A_x$ observables).}  Let us first determine the form of the first party's observables $A_x$. To this end, we concentrate on conditions (\ref{d3cond1}) which for $n=1$ and $r_{1,2}=1$ can be rewritten as
\begin{eqnarray}\label{Ayuda}
    \widetilde{A}_{0}\otimes B_0\otimes\overline{C}_1\ket{\psi}&\!\!\!=\!\!\!&\ket{\psi},\nonumber\\
    \widetilde{A}_{1}\otimes B_1\otimes\overline{C}_1\overline{C}_2\otimes\overline{D}\ket{\psi}&\!\!\!=\!\!\!&\ket{\psi},\nonumber\\
    \widetilde{A}_{2}\otimes B_2\otimes\overline{C}_1\overline{C}_2^{\dagger}\otimes\overline{D}^{\dagger}\ket{\psi}&\!\!\!=\!\!\!&\ket{\psi},
\end{eqnarray}
where $\overline{C}_i$ and $\overline{D}$ are short-hand notations for 
\begin{equation}
    \overline{C}_i=\bigotimes_{m=3}^{N_1+1}\left[C^{(m)}_0\right]^{r_{i,m}},\qquad 
    \overline{D}=\bigotimes_{m=N_1+2}^{N}\left[D^{(m)}_0\right]^{r_{2,m}},
\end{equation}
where $i=1,2$, and, finally,
\begin{equation}\label{tildeA}
    \widetilde{A}_k\equiv \widetilde{A}_k^{(1)}= \frac{\omega^{-k(k+1)}}{\sqrt{3}\,\lambda_1} \sum^{2}_{t=0} \w^{-tk}  A_{t}.
\end{equation}
Recall that in the case $d=3$, $\widetilde{A}_k^{(2)}=\widetilde{A}_{k}^{(1)\dagger}$. Moreover, since $\widetilde{\mathcal{G}}_{i,k}^{(2)}=\widetilde{\mathcal{G}}_{i,k}^{(1)\dagger}$, Eq. (\ref{d3cond1}) for $n=2$ gives another set of conditions, similar to (\ref{Ayuda}) but with all local operators being Hermitian-conjugated. By the very definition, $B_i$, $\overline{C}_i$ and $\overline{D}$ are unitary and satisfy $B_i^3=\overline{C}_i^3=\overline{D}^3=\mathbbm{1}$.

The above equations contain all three operators $\widetilde{A}_i$ $(i=0,1,2)$. Let us then concentrate on the first condition in (\ref{Ayuda}) and use the fact that $B_0$ and $\overline{C}_1$ are unitary to rewrite it as 
\begin{equation}
    \widetilde{A}_0\ket{\psi}=B_0^{\dagger}\otimes \overline{C}_1^{\dagger}\ket{\psi},
\end{equation}
which, taking into account that $B_0^3=\mathbbm{1}$ as well as 
$\overline{C}_1^3=\mathbbm{1}$, implies also that 
\begin{equation}
    \widetilde{A}_0^2\ket{\psi}=B_0\otimes \overline{C}_1\ket{\psi}.
\end{equation}
We can now use again the first condition in Eq. (\ref{Ayuda}) but with all local operators being "daggered" (recall that it follows from Eq. (\ref{d3cond1}) for $n=2$), which allows us to obtain 
$\widetilde{A}_0^2\ket{\psi}=\widetilde{A}_0^{\dagger}\ket{\psi}$. Since the reduced density matrix corresponding to the first subsystem of $\ket{\psi}$ is full rank, the latter is equivalent to the following relation 
\begin{equation}\label{raimat}
    \widetilde{A}_0^2=\widetilde{A}_0^{\dagger}.
\end{equation}
Using similar arguments one then shows that $\widetilde{A}_0$ is unitary, which together with (\ref{raimat}) implies that $\widetilde{A}_0^3=\mathbbm{1}$ and thus $\widetilde{A}_0$ is a proper quantum observable.

Employing then the second and the third relation in Eq. (\ref{Ayuda}), one can draw the same conclusions for the other two operators on Alice's side, $\widetilde{A}_1$ and $\widetilde{A}_2$. As a consequence, all three $\widetilde{A}_i$ are quantum observables; in particular, they satisfy 
\begin{equation}\label{Raimat2}
    \widetilde{A}_i^2=\widetilde{A}_i^{\dagger}\qquad (i=1,2,3).
\end{equation}

Let us now use (\ref{Raimat2}) to characterize $A_x$ observables. By substituting Eq. (\ref{tildeA}) into it one finds, after a bit of algebra, that the observables $A_x$ are related via the following formula:
\begin{eqnarray}\label{AnticommA}
    \{A_i,A_j\}=-\omega A_k^{\dagger},
\end{eqnarray}
where $i,j,k=0,1,2$ and $i\neq j\neq k$.
Using again Eq. (\ref{tildeA}) one can also derive similar 
relations for the tilted observables, 
\begin{equation}\label{rel1}
\{\widetilde{A}_i,\widetilde{A}_j\}=-\widetilde{A}_k^{\dagger}
\end{equation}
with $i,j,k=0,1,2$ such that $i\neq j\neq k$.

Importantly, equations (\ref{AnticommA}) and, analogously, (\ref{rel1}) were solved in Ref. \cite{kaniewski2018maximal}. 
In fact, it was proven there (cf. Fact \ref{dupablada} and Corollary \ref{Corollary} in Appendix \ref{Appendix}) that one can identify a qutrit Hilbert space in $\mathcal{H}_1$ in the sense that $\mathcal{H}_1=\mathbbm{C}^3\otimes\mathcal{H}_1'$
for some auxiliary Hilbert space $\mathcal{H}_1'$, and that there exists a unitary operation $U_1:\mathcal{H}_{1}\to \mathcal{H}_1$
such that [notice that the third observable $\widetilde{A}_2$ is obtained from the first two by using \eqref{rel1}]
\begin{equation}\label{similar}
    U_1\, \widetilde{A}_i\,U_1=XZ^i \otimes P_1^{(1)}+(XZ^i)^T\otimes P_2^{(1)}\qquad (i=0,1,2),
\end{equation}
where $P_i^{(1)}$ $(i=1,2)$ are two projectors
such that $P_1^{(1)}+P_2^{(1)}=\mathbbm{1}_1'$, where $\mathbbm{1}_1'$
is the indentity on $\mathcal{H}_1'$. There are thus two inequivalent sets of observables at the first site that give rise to the maximal violation of our Bell inequality: $XZ^i$ with $i=0,1,2$ and their transpositions. \ \\

\noindent \textbf{Step 2. ($B_y$ observables).} We can now move on to characterizing the $B_y$ observables. First, by combining the identities in (\ref{Ayuda}) with Eq. (\ref{rel1}) and then by using the fact that $\overline{C}_1$ and $\overline{C}_2$ commute as well as that $\widetilde{A}_i$ are unitary, one finds the following equations
\begin{eqnarray}
\{B_i,B_{j}\}\ket{\psi}&\!\!\!=\!\!\!&-B_{k}^{\dagger}\ket{\psi}
%
\end{eqnarray}
for all triples $i,j,k$ such that $i\neq j\neq k$.
By virtue of the fact that all the single-party reduced density matrices of $\ket{\psi}$ are full rank, these are equivalent to the following matrix equations
\begin{eqnarray}\label{AntiB}
\{B_0,B_{1}\}&\!\!\!=\!\!\!&-B_{2}^{\dagger},    \nonumber\\
\{B_0,B_{2}\}&\!\!\!=\!\!\!&-B_{1}^{\dagger},   \nonumber\\
\{B_1,B_2\}&\!\!\!=\!\!\!&-B_0^{\dagger},
\end{eqnarray}
and thus the $B_y$ observables satisfy analogous relations to $A_x$. 
This implies that $\mathcal{H}_2=\mathbbm{C}^3\otimes\mathcal{H}_2'$ for some auxiliary Hilbert space $\mathcal{H}_2'$, and there exists a unitary operation $U_2:\mathcal{H}_2\to\mathcal{H}_2$ such that
(cf. Fact \ref{dupablada} and Corollary \ref{Corollary})
%
%
%
\begin{equation}\label{Bopt2}
    U_2\, B_i\, U_2^{\dagger}=ZX^i\otimes P_1^{(2)}+(ZX^i)^T\otimes P_2^{(2)}.
\end{equation}
for $i=0,1,2$, where $P_1^{(2)}$ and $P_1^{(2)}$ are two orthogonal projectors such that $P_2^{(2)}+P_2^{(2)}=\mathbbm{1}_2'$, where $\mathbbm{1}_2'$ is the identity acting on $\mathcal{H}_2'$ [notice that as before the form of the third observable $B_2$ follows from \eqref{AntiB}].\\

\noindent\textbf{Step 3. ($C_z^{(i)}$ observables)}. Let us now move on to the $C_z^{(i)}$ observables that are measured by the observers numbered by $i=3,\ldots,N_1+1$, and consider the first equation in (\ref{Ayuda}) and the conditions that follow from (\ref{d3cond2}), which for our purposes we state as
\begin{equation}\label{maladupa}
    \widetilde{A}_{0}\otimes B_0\otimes \left[C_0^{(k)}\right]^{r_{1,k}}\otimes\overline{C}_{0,k}\ket{\psi}=\ket{\psi}
\end{equation}
%
%
and
\begin{equation}\label{wielkadupa}
    \widetilde{A}_{r_{1,k}}\otimes B_0^{r_{1,k}+r_{2,k}}\otimes C^{(k)}_1\otimes \overline{C}_{0,k}'\otimes  \overline{D}_k\ket{\psi}=\ket{\psi}
\end{equation}
with $k=3,\ldots,N_1+1$, and
\begin{equation}
  \overline{C}_{0,k}=\bigotimes_{\substack{m=3\\m\neq k}}^{N_1}\left[C^{(m)}_0\right]^{r_{1,m}},\qquad \overline{C}_{0,k}'=\bigotimes_{\substack{m=3\\m\neq k}}^{N_1}\left[C^{(m)}_0\right]^{r_{1,m}+r_{k,m}},\qquad \overline{D}_k=\bigotimes_{m=N_1+1}^{N}\left[D^{(m)}_0\right]^{r_{k,m}}.
\end{equation}

Importantly, $r_{1,k}\neq 0$ for any $k=3,\ldots,N_1+1$, and hence
all equations in (\ref{wielkadupa}) contain either $\widetilde{A}_1$ 
or $\widetilde{A}_2$. Let us then exploit the fact that all local operators in both Eqs. (\ref{maladupa}) and (\ref{wielkadupa}) are unitary and therefore these equations can be rewritten as 
\begin{eqnarray}
    \left[C_0^{(k)}\right]^{r_{1,k}}\ket{\psi}&\!\!\!=\!\!\!&\widetilde{A}_0^{\dagger}\otimes B_0^{\dagger}\otimes \overline{C}_{0,k}^{\dagger}\ket{\psi},\nonumber\\
    C_1^{(k)}\ket{\psi}&\!\!\!=\!\!\!&\widetilde{A}_{1}^{\dagger}\otimes B_0^{-(r_{1,k}+r_{2,k})} \otimes \left[\overline{C}_{0,k}'\right]^{\dagger}\otimes \overline{D}_k^{\dagger}\ket{\psi}.
\end{eqnarray}
Crucially, $\overline{C}_{0,k}$, $\overline{C}_{0,k}'$ commute and therefore we deduce that 
\begin{equation}
  \left \{\left[C_0^{(k)}\right]^{r_{1,k}},C_1^{(k)}\right\}\ket{\psi}=\{\widetilde{A}_{0},\widetilde{A}_{1}\}^{\dagger}\otimes B_0^{\lambda_k}\otimes \overline{C}_{0,k}^{\dagger}\left[\overline{C}_{0,k}'\right]^{\dagger}\otimes \overline{D}_k^{\dagger}\ket{\psi},
\end{equation}
where for simplicity we denoted $\lambda_k=-(1+r_{1,k}+r_{2,k})$. In a fully analogous way we can derive 
\begin{equation}
  \left  \{\left[C_0^{(k)}\right]^{r_{1,k}},C_1^{(k)}\right\}^{\dagger}\ket{\psi}=\{\widetilde{A}_{0},\widetilde{A}_{1}\}\otimes B_0^{-\lambda_k}\otimes \overline{C}_{0,k}\,\overline{C}_{0,k}'\otimes\overline{D}_k\ket{\psi}.
\end{equation}
Both these conditions when combined with Eq. \eqref{rel1} allow us to conclude that 
\begin{equation}\label{anticommutator15}
    \left  \{\left[C_0^{(k)}\right]^{r_{1,k}},C_1^{(k)}\right\}^{\dagger}\left  \{\left[C_0^{(k)}\right]^{r_{1,k}},C_1^{(k)}\right\}=
    \left  \{\left[C_0^{(k)}\right]^{r_{1,k}},C_1^{(k)}\right\}\left  \{\left[C_0^{(k)}\right]^{r_{1,k}},C_1^{(k)}\right\}^{\dagger}=\mathbbm{1}_k,
\end{equation}
i.e., the above anticommutator is unitary.
We can therefore use Fact \ref{dupablada} and Corollary \ref{Corollary}
(see Appendix \ref{Appendix}) which say that for any $k=3,\ldots,N_1+1$, $\mathcal{H}_k=\mathbbm{C}^3\otimes\mathcal{H}_k'$ with
$\mathcal{H}_k'$ being some auxiliary Hilbert space of unknown dimension, as well as that there exist unitary operations $U_k$ such that
\begin{equation}\label{Copt1}
    U_k\,\left[C_0^{(k)}\right]^{r_{1,k}}\,U_k^{\dagger}=Z^{r_{1,k}}\otimes\mathbbm{1}_k',
\end{equation}
and
\begin{equation}\label{Copt2}
    U_k\,C_1^{(k)}\,U_k^{\dagger}=Z^{r_{1,k}}X\otimes {P_1^{(k)}}+
    (Z^{r_{1,k}}X)^T\otimes {P_2^{(k)}},
\end{equation}
where $P_1^{(k)}+P_2^{(k)}=\mathbbm{1}_k'$. \\

\noindent\textbf{Step 4. ($D_w^{(i)}$ observables)}. Let us finally focus on the $D$ observables. We first consider all vertices $i\in\{N_2+2,\ldots,N\}$ that are connected to the second vertex. For them $r_{2,k}\neq 0$ and therefore we have from Eq. (\ref{d3cond3}),
\begin{equation}\label{D1}
    B_0^{r_{2,k}}\otimes \widetilde{C}_{0,k}\otimes \overline{D}_{0,k}'\otimes D_{1}^{(k)}\,\ket{\psi}=\ket{\psi},
\end{equation}
where
\begin{equation}
    \widetilde{C}_{0,k}=\bigotimes_{m=3}^{N_1}\left[C^{(m)}_0\right]^{r_{k,m}},
    \qquad  \overline{D}_{0,k}'=\bigotimes_{\substack{i=N_1+1\\i\neq k}}^N \left[D_0^{(i)}\right]^{r_{k,i}}.
\end{equation}
At the same time, Eq. (\ref{d3cond1}) for $k=1$ gives
\begin{equation}
\widetilde{A}_{r_{1,2}}\otimes B_1\otimes \overline{C}_1\overline{C}_2 \otimes \left[D_0^{(k)}\right]^{r_{2,k}}\otimes \overline{D}_{0,k}\,\ket{\psi}=\ket{\psi}
\end{equation}
where
\begin{equation}\label{D2}
    \overline{D}_{0,k}=\bigotimes_{\substack{i=N_1+1\\i\neq k}}^N \left[D_0^{(i)}\right]^{r_{2,i}}.
\end{equation}
We then rewrite both Eqs. (\ref{D1}) and (\ref{D2}) as
\begin{eqnarray}
    D_1^{(k)}\ket{\psi}&\!\!\!=\!\!\!&B_0^{-r_{2,k}}\otimes \widetilde{C}_{0,k}^{\dagger}\otimes \left[\overline{D}'_{0,k}\right]^{\dagger}\ket{\psi},\nonumber\\ \left[D_0^{(k)}\right]^{r_{2,k}}\ket{\psi}&\!\!\!=\!\!\!&\widetilde{A}_{1}^{\dagger}\otimes B_1^{\dagger}\otimes \overline{C}_1^{\dagger}\overline{C}_2^{\dagger}\otimes\overline{D}_{0,k}^{\dagger}\ket{\psi}
\end{eqnarray}
Since as already proven, the anticommutator of $B_0^{-r_{2,k}}$ and $B_1$ is unitary for any $k$ such that $r_{2,k}\neq 0$, the above equations imply that for all $k=N_1+2,\ldots,N$ 
for which $r_{2,k}\neq 0$, the anticommutator of 
$D_1^{(k)}$ and $[D_0^{(k)}]^{r_{2,k}}$ is unitary too. 

We can now move on to those vertices $i\in\{N_1+2,\ldots,N\}$ that are connected to the remaining neighbours of the first vertex. In this case we proceed in the same way as above, however, we now combine the conditions (\ref{d3cond2}) and (\ref{d3cond3}) as well as we employ the forms of the $C_z^{(i)}$ operators given in Eqs. (\ref{Copt1}) and (\ref{Copt2}) to observe that for any site $k$ which is connected to a neighbour $m$ of the first vertex the anticommutator of $D_{1}^{(k)}$
and $[D_{0}^{(k)}]^{r_{m,k}}$ is unitary and therefore $D_{0/1}^{(k)}$
satisfy the assumptions of Fact \ref{dupablada} in Appendix \ref{Appendix}.

Let us finally consider the remaining vertices that are not neighbours of the first vertex. For each of them we can prove that the anticommutator of the local observables $D_{0/1}^{(k)}$ or powers thereof is unitary in a recursive way starting from vertices connected to those that are connected to the neighbours of the first vertex and employing the relations \eqref{d3cond3}. Step by step we can prove the same statement for all $D$ sites exploiting the fact that the graph is connected and therefore for each vertex there is a path connecting it with any other vertex in the graph. 

We thus conclude that for all vertices $k=N_1+2,\ldots,N$ the local Hilbert is $\mathcal{H}_{k}=\mathbbm{C}^3\otimes \mathcal{H}_k'$
for some finite-dimensional $\mathcal{H}_k'$ and that there exists a unitary $U_k$ such that (cf. Fact \ref{dupablada} and Corollary \ref{Corollary} in Appendix \ref{Appendix})
\begin{equation}\label{Dopt1}
    U_k\, D_0^{(k)}\, U_k^{\dagger}=Z\otimes \mathbbm{1}_k'
\end{equation}
and
\begin{equation}\label{Dopt2}
    U_k\, D_1^{(k)}\, U_k^{\dagger}=X\otimes P_1^{(k)}+X^T\otimes P_2^{(k)}.
\end{equation}

\ \\

\noindent\textbf{The state.} Having determined the form of all local observables we can now move on to proving the self-testing statement for the state. After substituting the above observables, the "rotated" Bell operator corresponding to the Bell inequality which is maximally violated can be expressed as  
\begin{equation}\label{BigBell}
    U\,\mathcal{B}_{\mathcal{G}}\,U^{\dagger}=\sum_{m_1,\ldots,m_N=0}^1 B_{\boldsymbol{m}}\otimes P_{m_1}^{(1)}\otimes\ldots\otimes P_{m_N}^{(N)},
\end{equation}
where $U=U_1\otimes\ldots\otimes U_N$ and $P_{m_i}^{(i)}$ are projections introduced above that satisfy $P^{(i)}_1P^{(i)}_2=0$ for any site $i=1,\ldots,N$, $B_{\boldsymbol{m}}$
with $\boldsymbol{m}:=m_1\ldots m_N$, where $m_i=0,1$, are $N$-qutrit Bell operators obtained from 
%
%
\begin{equation}\label{smallBell}
    B=\mathcal{G}_{1,0}
 +\sum_{l=1}^{2}c_{1,l}\,\mathcal{G}_{1,l}
+\sum_{l=3}^{N_1+1}c_{2,l}\,\mathcal{G}_{2,l}+
    \sum_{l=N_1+2}^N\mathcal{G}_{3,l}+\mathrm{h.c.},
\end{equation}
through the application of the identity map ($m_i=0$) or the transposition map ($m_i=1$) to the observables appearing at site $i$.
%
%
Here, $\mathcal{G}_{a,b}$ are the stabilizing operators 
of the graph state $\ket{G}$ defined in 
Eqs. (\ref{StabOp2}) for $n=1$ and $d=3$, which for completeness we restate here as
\begin{equation}\label{GOp1}
    \mathcal{G}_{1,0}=X_1\otimes Z_2\otimes\bigotimes_{i=3}^{N_1+1}Z_i^{r_{1,i}},
\end{equation}
\begin{equation}\label{GOp2}
    \mathcal{G}_{1,k}=\left(XZ^k\right)_1
    \otimes \left(ZX^k\right)_2\otimes\bigotimes_{i=3}^{N_1+1}Z_i^{r_{1,i}+kr_{2,i}}\otimes\bigotimes_{i=N_1+2}^{N}Z_i^{kr_{2,i}},
\end{equation}
with $k=1,2$,
\begin{equation}\label{GOp3}
    \mathcal{G}_{2,k}=(XZ^{r_{1,k}})_1
    \otimes Z_2^{r_{1,k}+r_{2,k}}\otimes\bigotimes_{i=3}^{k-1}
    Z^{r_{1,i}+r_{k,i}}_i\otimes\left(Z^{r_{1,k}}X\right)_k\otimes\bigotimes_{i=k+1}^{N_1+1}
    Z^{r_{1,i}+r_{k,i}}_i\otimes\bigotimes_{i=N_2+2}^{N}Z^{r_{k,i}}_i
\end{equation}
with $k=3,\ldots,N_1+1$
\begin{equation}\label{GOp4}
    \mathcal{G}_{3,k}=
    Z^{r_{2,k}}_2\otimes\bigotimes_{i=3}^{N_1+1}
    Z^{r_{k,i}}_i\otimes\bigotimes_{i=N_1+2}^{k-1}Z_i^{r_{k,i}}\otimes
    X_k\otimes\bigotimes_{i=k+1}^{N}Z_i^{r_{k,i}},
\end{equation}
with $k=N_1+2,\ldots,N$. The subscripts were added to $X$ and $Z$ to denote the site at which these operators act; recall also that we fixed $r_{1,2}=1$.

The formula (\ref{BigBell}) takes into account the fact that at each site we have two choices of measurements, with and without the transposition. Thus, the Bell operator is composed of $2^N$ $N$-qutrit Bell operators. 
For instance, for $m_1=\ldots=m_N=0$ no partial transposition is applied to $B$ and therefore $B_{0\ldots 0}\equiv B$, whereas for $m_1=\ldots=m_N=1$ the partial transposition is applied to every site and hence $B_{1\ldots 1}=B^T$, where $T$ stands for the global transposition. 

In order to find the form of the state maximally violating our inequality we now determine the eigenvector(s) of the Bell operator $\mathcal{B}_{\mathcal{G}}$ corresponding its maximal eigenvalue which is $2(N-N_1+d-1)$ [cf. Eq. (\ref{MaxQVal})]. 
To this end, let us focus on the $N$-qutrit operators $B_{\boldsymbol{m}}$ and prove that the latter number is an eigenvalue of only two of them, $B$ and $B^T$, which correspond to the cases $m_1=m_2=\ldots=m_N=0,1$, whereas the eigenvalues of the remaining operators are all lower. 

Clearly, $B$ is composed of the stabilizing operators of the graph state $\ket{G}$ and therefore its maximal eigenvalue conincides with the maximal quantum violation of the inequality which is $2(N-N_1+d-1)$. The same applies to $B^T$ because the transposition does not change the eigenvalues and the graph state is real.

Let us then move on to the remaining cases, i.e., $m_i$ are not all equal. We will show that in all those 
$2^{N}-2$ cases the $B_{\boldsymbol{m}}$ operators
have eigenvalues lower than $2(N-N_1+2)$ because for all those cases one can pick a few stabilizing operators $\mathcal{G}_{a,b}$ whose partial transpositions cannot stabilize a common pure state anymore. For further benefits let us denote by $\mathcal{G}_{a,b}^{\boldsymbol{m}}$ the stabilizing operators which are partially transposed with respect to those subsystems $i$ for which $m_i=1$. We divide the proof into three parts corresponding to three cases: (i) $m_1=m_2=0$, (ii)$ m_1=m_2=1$ and (iii) $m_1=0$, $m_2=1$ or $m_1=1$, $m_2=0$, and also a few sub-cases. 
\begin{itemize}
    \item The first one assumes that either $m_1=1$ and $m_2=0$ or $m_1=0$ and $m_2=1$, i.e., we take the transposed observables at the first or the second site, but not both at the same time. For simplicity let us then fix $m_1=1$ and $m_2=0$. We consider three operators $\mathcal{G}_{1,i}^{T_1}$ with $i=0,1,2$, where $T_1$ is the transposition applied to the observables at the first site.  It is not difficult to observe that using the explicit forms of the
    stabilizing operators [cf. Eqs. (\ref{GOp1}) and (\ref{GOp2})] and including the transposition at the first site, one obtains
    \begin{equation}
        \mathcal{G}_{1,0}^{T_1}\mathcal{G}_{1,1}^{T_1}\mathcal{G}_{1,2}^{T_1}=[X^T(XZ)^T(XZ^2)^T]_1\otimes [Z\,ZX\,ZX^2]_2,
    \end{equation}
    where we also used the fact that the products of the observables at the remaining sites amounts to identity. Using then the fact that $ZX=\omega XZ$, the above simplifies to
    \begin{equation}
       \mathcal{G}_{1,0}^{T_1}\mathcal{G}_{1,1}^{T_1}\mathcal{G}_{1,2}^{T_1}=\omega \mathbbm{1}.
    \end{equation}
    This simple fact precludes that there exists a common eigenvector of 
    $\mathcal{G}_{1,i}^{T_1}$ $(i=1,2,3)$ with eigenvalue one.
    
    \item Next, we consider the case when the observables at the first two sites are not transposed, i.e., $m_1=m_2=0$. There thus exists $i\neq 1,2$ such that $m_i=1$. Let us first assume that this particular vertex belongs to $i\in\{3,\ldots,N_1+1\}$, i.e., we take the transposed observables for this site. We then consider two operators $\mathcal{G}_{1,0}$ and $\mathcal{G}_{2,i}$. Notice then that the first of these operators has the $Z$ observable at site $i$ because $i\in \mathcal{N}_1$, i.e., it is connected to the first vertex, whereas the second one has $Z^{r_{1,i}}X$ at this position. At the remaining positions different than the first two they have only $Z$ observable or the identity which do not feel the action of transposition. All this means that in this case 
    $\mathcal{G}_{1,0}^{\boldsymbol{m}}=\mathcal{G}_{1,0}$ and 
    $\mathcal{G}_{2,i}^{\boldsymbol{m}}=\mathcal{G}_{2,i}^{T_i}$. Due to the fact that the transposition at site $i$ modifies
    $X$ appearing in $\mathcal{G}_{2,i}$ to $X^{\dagger}$, the operators 
    $\mathcal{G}_{1,0}$ and $\mathcal{G}_{2,i}^{T_i}$ do not commute (recall that by the very definition the stabilizing operators without the transposition commute). By virtue of Fact \ref{fact:comut} stated in Appendix \ref{App0} this implies that $\mathcal{G}_{1,0}$ and 
    $\mathcal{G}_{2,i}^{T_i}$ do not stabilize a common pure state.\\

    Let us now move on to the second sub-case in which $m_i=1$ for any $i\in \mathcal{N}_1$ and there exist $i\in\{N_1+2,\ldots,N\}$ such that $m_i=2$. Since the graph is connected there exist another vertex $j\neq 1,i$ which is connected to $i$. Analogously to the previous case, we consider two operators: $\mathcal{G}_{3,i}^{\boldsymbol{m}}$ and one of $\mathcal{G}_{a,b}^{\boldsymbol{m}}$, where the choice of the latter operator is dictated by the choice of the vertex $j$ which $i$ is connected to: for $j=2$ we take $\mathcal{G}_{1,1}^{\boldsymbol{m}}$; for $j\in\{3,\ldots,N_1+1\}$ we take $\mathcal{G}_{2,j}^{\boldsymbol{m}}$; finally, for $j\in\{N_1+2,\ldots,N\}$ we take $\mathcal{G}_{3,j}^{\boldsymbol{m}}$. 
    
    Now, $\mathcal{G}_{3,i}^{\boldsymbol{m}}$ has the $X$ operator at site $i$ and the $Z$ operator at the remaining "$D$" sites, whereas all the other operators $\mathcal{G}_{a,b}^{\boldsymbol{m}}$ for $a=1,2,3$ and $b\neq i$ listed above have only either the $Z$ operator or the identity at all "$D$" sites. Thus, $\mathcal{G}_{a,b}^{\boldsymbol{m}}=\mathcal{G}_{a,b}$ for any $a=1,2,3$ and $b\neq i$ and any sequence $\boldsymbol{m}$ in which $m_l=1$ for $l=1,\ldots,N_1+1$, and   $\mathcal{G}_{3,i}^{\boldsymbol{m}}=\mathcal{G}_{3,i}^{T_i}$. Now, it clearly follows that $\mathcal{G}_{3,i}^{\boldsymbol{m}}$ does not commute with the chosen $\mathcal{G}_{a,b}$ because the transposition at site $i$ changes the $X$ operator to $X^2$ and because, by the very definition, 
    $\mathcal{G}_{3,i}$ (without the transposition) commutes with any other $\mathcal{G}_{a,b}$. As before this implies that 
    $\mathcal{G}_{3,i}^{\boldsymbol{m}}\mathcal{G}_{a,b}=\omega^q\mathcal{G}_{a,b}\mathcal{G}_{3,i}^{\boldsymbol{m}}$ for some $q=1,2$, and therefore these two operators cannot stabilize a common pure state [cf. Fact \ref{fact:comut} in Appendix \ref{App0}].
    
    \item The last case to consider is when $m_1=m_2=1$; the remaining $m_i$ can take arbitrary values except for being all equal to one, which corresponds to the already-considered case of all observables being transposed. Here we can use the fact that $\mathcal{G}_{a,b}^{\boldsymbol{m}}$ for all $a,b$ stabilize the graph state $\ket{G}$ if and only if  $[\mathcal{G}_{a,b}^{\boldsymbol{m}}]^{T}$ does, where $T$ is the global transposition. We can thus apply the global transposition to all the operators $\mathcal{G}_{a,b}^{\boldsymbol{m}}$ and consider again the case when $m_0=m_2=0$ and there is some $i\neq 1,2$ such that $m_i = 1$, which has already been considered above.
\end{itemize}

Knowing that among all the $B_{\boldsymbol{m}}$
operators only $B$ and $B^T$ give rise to the maximal quantum violation of the Bell inequality corresponding to the considered graph, 
we can determine the form of the state $\ket{\psi}$ maximally violating the inequality. Due to the fact that each local Hilbert space decomposes as $\mathcal{H}_k=\mathbbm{C}^3\otimes \mathcal{H}_k'$ we can write the state as
\begin{equation}\label{state}
   \ket{\psi}=\sum_{i_1,\ldots,i_N}\ket{\psi_{i_1,\ldots,i_N}}\otimes\ket{i_1}_{1}\otimes\ldots\otimes\ket{i_N}_N,
\end{equation}
where $\ket{\psi'}=(U_1\otimes\ldots\otimes U_N)\ket{\psi}$, $\ket{\psi_{i_1,\ldots,i_N}}$ are some vectors from $(\mathbbm{C}^3)^{\otimes N}$ and the local bases $\ket{i_k}$ are the eigenbases of the projectors $P^{(k)}_{m_k}$. The fact that $\ket{\psi}$ achieves the maximal quantum value of the inequality, $\beta_Q=2(N-N_1+2)$, means that the following identity
\begin{equation}\label{rownanie}
    \mathcal{B}_{\mathcal{G}}\ket{\psi}=2(N-N_1+2)\ket{\psi}
\end{equation}
holds true. Plugging Eqs. (\ref{state}) and (\ref{BigBell}) into the above equation one finds that it is satisfied iff for every sequence
$\boldsymbol{m}$,
\begin{equation}\label{rownanie2}
    B_{\boldsymbol{m}}\ket{\psi_{i_1,\ldots,i_N}}=2(N-N_1+2)\ket{\psi_{i_1,\ldots,i_N}},
\end{equation}
holds true for all those sequences $i_1,\ldots,i_N$ for which 
the local vectors $\ket{i}_k$ at site $k$ are the eigenvectors of the operator $P_{m_k}^{(k)}$. As already discussed above, this condition can be met for only two of these operators, $B$ and $B^{T}$. Moreover, the stabilizing operators that $B$ (and thus also $B^{T}$) are composed of stabilize a unique state, which is the graph state $\ket{G}$. Consequently,  $\ket{\psi_{i_1,\ldots,i_N}}=\ket{G}$ for any sequence $i_1,\ldots,i_N$ for which the corresponding local 
vectors are the eigenvectors of $P_{0}^{(k)}$ (or $P_{1}^{(k)}$ in the case of $B^T$).

On the other hand, we showed that the eigenvalues of the remaining operators $B_{\boldsymbol{m}}$ are lower than the maximal violation of the Bell inequality and thus in all those cases Eq. (\ref{rownanie2}) can be satisfied iff the corresponding vectors vanish, $\ket{\psi_{i_1,\ldots,i_N}}=0$.
Taking all this into account, we conclude that the state 
$\ket{\psi'}$ has the following form
\begin{equation}
    (U_1\otimes\ldots\otimes U_N)\ket{\psi}=\ket{\psi_{\mathcal{G}}}\otimes\ket{\varphi},
\end{equation}
where $\ket{\varphi}$ is some state from the auxiliary Hilbert spaces $\mathcal{H}_1'\otimes\ldots\otimes\mathcal{H}_N'$ that satisfies 
\begin{equation}
 \left(P_{i}^{(1)}\otimes\ldots\otimes P_{i}^{(N)}\right)\ket{\varphi}=\ket{\varphi}\qquad (i=0,1).
\end{equation}
This completes the proof.
\end{proof}

\section{Conclusions and outlook}
\label{Sec:Conclusion}

In this work we introduced a family of Bell expressions whose maximal quantum values are achieved by graph states of arbitrary prime local dimension. While at the moment we are unable to compute their maximal classical values, we believe the corresponding Bell inequalities are all nontrivial. This belief is supported by a few examples of Bell inequalities for which the classical bound was found numerically, and the fact that in the particular case of qutrit states they enable self-testing of all graph states. We thus introduced a broad class of Bell inequalities that can be used for testing non-locality of many interesting and relevant multipartite states, including the absolutely maximally entangled states. Moreover, in the particular case of many-qutrit systems our inequalities can also be employed to self-test the graph states, in particular the four-qutrit absolutely maximally entangled state.

There is a few possible directions for further research that are inspired by our work:
\begin{itemize}
    \item First of all, as far as implementations of self-testing are concerned it is a problem of a high relevance to understand how robust our self-testing statements are against noises and experimental imperfections.
    
    \item Another possible direction that is related to the possibility of experimental implementations of self-testing is to find Bell inequalities maximally violated by graph states that require performing the minimal number of two measurement per observer to self-test the state. For instance, for the GHZ state such a Bell inequality \cite{Augusiak2019} and a self-testing scheme \cite{sarkar2019self} based on the maximal violation of this inequality were introduced recently; this inequality is based, however, on a slightly different construction which is not directly related to the stabilizer formalism used by us here.

\item Third, it is interesting to explore whether one can derive 
self-testing statements based on the maximal violation of our inequalities for higher prime dimensions $d\geq 0$. While it is already known (cf. Ref. \cite{kaniewski2018maximal}) that these inequalities do not serve the purpose as far as quantum observables are concerned because there exist many different choices of them that are not unitarily equivalent, whether they enable self-testing of graph states remains open. In other words, it is unclear whether the given graph state 
is the only one (up to the above equivalences) that meets the necessary and sufficient conditions for the maximal quantum violation of the corresponding Bell inequality stemming from the sum-of-squares decomposition.

\item The fourth possible direction is to generalize our construction so that it allows for designing Bell inequalities that are maximally violated by other classes of states such as for instance the hyper-graph states \cite{Rossi_2013} (see also Ref. \cite{PhysRevLett.116.070401} in this context).

\item Last but not least, one can also explore the possibility of self-testing of genuinely entangled subspaces within the stabilizer formalism in Hilbert spaces of arbitrary prime local dimension along the lines of Refs. \cite{Subspaces1,Subspaces2}.

\end{itemize}

\section*{Acknowledgments}

We acknowledge the VERIqTAS project funded within the QuantERA II Programme that has received funding from the European Union’s Horizon 2020 research and innovation programme under Grant Agreement No 101017733 and the Polish National Science Center. F. B. is supported by the Alexander von Humboldt foundation.

\appendix

\section{A few facts}
\label{App0}

\begin{fakt}\label{fact:unitary}
    Consider the generalized Pauli matrices defined through the following formulas
\begin{equation}
    X\ket{i}=\ket{i+1},\qquad Z\ket{i}=\omega^i\ket{i},
\end{equation}    
where $\ket{i}$ $(i=0,\ldots,d-1)$ are the elements of the standard basis of $\mathbbm{C}^d$. There are no complex numbers $\alpha,\beta\neq 0$ for which $\alpha X+\beta Z$ is unitary. 
\end{fakt}
\begin{proof}The proof is elementary. We first expand
\begin{equation}\label{TintoVerano}
    (\alpha X+\beta Z)^{\dagger}(\alpha X+\beta Z)=(|\alpha|^2+|\beta|^2)\mathbbm{1}+\alpha^{*}\beta X^{\dagger}Z+\beta^{*}\alpha Z^{\dagger}X.
\end{equation}
Let us then show that for any $d\geq 3$, the operators $X^{\dagger}Z$ and $Z^{\dagger}X$ are linearly independent. To this end, we assume that 
$X^{\dagger}Z$ and $Z^{\dagger}X$ are linearly dependent and thus $X^{\dagger}Z=\eta Z^{\dagger}X$ for some 
$\eta\in\mathbbm{C}$. By using the fact that $ZX=\omega XZ$, we can rewrite this equation as $X^{\dagger}Z=\eta \omega^{d-1}XZ^{\dagger}$, which, taken into account the fact that $X$ and $Z$ are unitary further rewrites as $Z^2=\eta\omega^{d-1}X^2$ which for $d\geq 3$ is satisfied iff $\eta=0$. 

It now follows that the expression (\ref{TintoVerano}) equals $\mathbbm{1}$ if and only if $\alpha$ or $\beta$ vanishes. 
This completes the proof. 
\end{proof}
Let us notice that the above fact fails to be true for $d=2$ because in this case $ZX=-XZ$ and therefore $XZ$ and $ZX$ are linearly dependent, which makes it possible to find $\alpha,\beta$
such that $\alpha X+\beta Z$ is unitary. In fact, any pair of real positive numbers obeying $\alpha^2+\beta^2=1$ makes this matrix unitary.

Let us finally provide a proof of the properties (\ref{prop1}) and (\ref{prop2}). For this purpose we recall $\widetilde{A}^{(n)}_k$ to be given by
\begin{equation}\label{dupa2}
\widetilde{A}^{(n)}_k:= \frac{\omega^{-nk(k+1)}}{\sqrt{d}\,\lambda_n} \sum^{d-1}_{t=0} \w^{-ntk}  A_{t}^n,
\end{equation}
where $A_t$ are unitary observables.

\begin{fakt}\label{fact:properties}
Consider the following matrices
\begin{equation}\label{dupa2}
\widetilde{A}^{(n)}_k:= \frac{\omega^{-nk(k+1)}}{\sqrt{d}\,\lambda_n} \sum^{d-1}_{t=0} \w^{-ntk}  A_{t}^n,
\end{equation}
where $A_t$ are unitary observables.
For any $n=0,\ldots,d-1$, the following identity holds true:
\begin{equation}\label{app:prop2}
\sum^{d-1}_{k=0}  \widetilde{A}^{(d-n)}_k \widetilde{A}^{(n)}_k =\sum^{d-1}_{k=0}  \left[\widetilde{A}^{(n)}_k\right]^{\dagger} \widetilde{A}^{(n)}_k = d \I.  
\end{equation}
\end{fakt}
\begin{proof}After plugging Eq. (\ref{dupa2}) into Eq. (\ref{app:prop2}),
one obtains
\begin{equation}
    \sum^{d-1}_{k=0}  \widetilde{A}^{(d-n)}_k \widetilde{A}^{(n)}_k=\frac{1}{d|\lambda_n|^2}\sum_{s,t=0}^{d-1}\sum_{k=0}^{d-1}\omega^{nk(s-t)}A_{s}^{-n}A_{t}^{n}.
\end{equation}
Employing then the following identity
\begin{equation}
    \sum_{k=0}^{d-1}\omega^{nk(s-t)}=d\delta_{s,t}
\end{equation}
and the fact that $|\lambda_n|^2=1$ for any $n$,
one directly arrives at Eq. (\ref{app:prop2}), which completes the proof.
\end{proof}

\begin{fakt}\label{fact:comut}
Consider two $N$-qudit operators $S_1$ and $S_2$ which are
$N$-fold tensor products of $X^iZ^j$ with $i,j=0,\ldots,d-1$ with prime $d$. Assume also that $S_1^d=S_2^d=\mathbbm{1}$. If $[S_1,S_2]\neq 0$,
then they cannot stabilize a common pure state in; in other words, no nonzero $\ket{\psi}\in(\mathbbm{C}^d)^{\otimes N}$ exists such that $S_i\ket{\psi}=\ket{\psi}$ for $i=1,2$.
\end{fakt}
\begin{proof}Let us first notice that the Weyl-Heisenberg matrices $W_{i,j}=X^iZ^j$
satisfy the following commutation relations $W_{i,j}W_{k,l}=\omega^{f(i,j,k,l)}W_{k,l}W_{i,j}$ with $f:\{0,\ldots,d-1\}^4\to \{0,\ldots,d-1\}$, and thus there exists $q=\{1,\ldots,d-1\}$ such that
\begin{equation}\label{comutrelation}
S_1S_2=\omega^q S_2S_1    \qquad (q=1,\ldots,d-1),
\end{equation}
where $q\neq 0$ due to the assumption that $S_i$ do not commute. 

Now, let us assume that $S_i$ stabilize a common pure state, $S_i\ket{\psi}=\ket{\psi}$ for $i=1,2$. Then, the relation 
(\ref{comutrelation}) implies $\ket{\psi}=\omega^q\ket{\psi}$ which is satisfied iff $\ket{\psi}=0$, which leads to a contradiction. This ends the proof.
\end{proof}

\section{Characterization of observables}
\label{Appendix}

The following proposition was proven in Appendix B of Ref. \cite{kaniewski2018maximal}.
\begin{fakt}\label{dupablada}
Let $R_0$ and $R_1$ acting on some finite-dimensional Hilbert space $\mathcal{B}$ be unitary operators satisfying $R_0^3=R_1^3=\mathbbm{1}$. If the anticommutator
$\{R_0,R_1\}$ is unitary, then $\mathcal{H}=\mathbbm{C}^3\otimes\mathcal{H}'$ for some Hilbert space $\mathcal{H}'$ and there exists a unitary $U:\mathcal{H}\to \mathbbm{C}^3\otimes\mathcal{H}'$ such that
\begin{eqnarray}\label{app:Bopi}
    UR_0U^{\dagger}&\!\!\!=\!\!\!&X\otimes Q+X\otimes Q^{\perp}=X\otimes\mathbbm{1}',\nonumber\\
    UR_1U^{\dagger}&\!\!\!=\!\!\!&X^2Z\otimes Q+Z^2\otimes Q^{\perp},
\end{eqnarray}
where $Q$ and $Q^{\perp}$ are orthogonal projections satisfying $Q+Q^{\perp}=\mathbbm{1}'$ and $\mathbbm{1}'$ stands for the identity acting on $\mathcal{H}'$.
\end{fakt}

Based on the above fact let us now show demonstrate that for each of the subsets of observables $A_x$, $B_y$, $C_z^{(i)}$ and $D_{w}^{(i)}$
there exist local unitary operations bringing them to 
the forms used in Eqs. \eqref{similar}, \eqref{Bopt2}, \eqref{Copt1} and \eqref{Copt2}, and finally, 
\eqref{Dopt1} and \eqref{Dopt2}.

%
\begin{cor}\label{Corollary}
The following statements can be verified by a direct check: 
\begin{itemize}
    \item \textbf{($A_x$ observables)} By using $\mathcal{U}_1=F^{\dagger}V_1 F\otimes Q_1+F^{\dagger}V_1^* V_2 F\otimes Q_2$, where 
    $F$, $V_1$ and $V_2$ are unitary operations given by 
    \begin{equation}
    F=
    \frac{1}{\sqrt{3}}\left(
    \begin{array}{ccc}
        1 &  1  & 1 \\
        1 &  \omega  & \omega^2  \\
        1 &   \omega^2  & \omega
    \end{array}
    \right), \qquad V_1=
    \left(\begin{array}{ccc}
        1 &  0  & 0 \\
        0 &  1  & 0  \\
        0 &   0  & \omega
    \end{array}\right),\qquad V_2=
    \left(\begin{array}{ccc}
        1 &  0  & 0 \\
        0 &  0  & 1  \\
        0 &   1  & 0
    \end{array}\right),
\end{equation}
one can bring the observables in Eq. (\ref{app:Bopi}) into the following form used in Eq. \eqref{similar}:
\begin{eqnarray}
      &&X\otimes Q+X^T\otimes Q^{\perp},\nonumber\\
      &&XZ\otimes Q+(XZ)^T\otimes Q^{\perp}.
\end{eqnarray}

\item \textbf{($B_y$ and $C_z^{(i)}$ observables)} By using $\mathcal{U}_2=V_3 F\otimes Q+(V_1V_3)^*F\otimes Q^{\perp}$, and relabeling $Q\leftrightarrow Q^{\perp}$ one brings the observables (\ref{app:Bopi}) to those in Eq. \eqref{Bopt2}, that is,
\begin{eqnarray}\label{Bobs}
      &&Z\otimes Q+Z\otimes Q^{\perp}=Z\otimes\mathbbm{1},\nonumber\\
      &&ZX\otimes Q+(ZX)^T\otimes Q^{\perp},
\end{eqnarray}
where 
\begin{equation}
    V_3=
    \left(\begin{array}{ccc}
        1 &  0  & 0 \\
        0 &  \omega^2  & 0  \\
        0 &   0  & \omega^2
    \end{array}\right).
\end{equation}
Then, by applying $\mathcal{U}_3=(V_1V_3)^*\otimes Q+V_2'\otimes Q^{\perp}$ operation we can bring (\ref{Bobs}) to 
\begin{eqnarray}\label{Charlie}
      &&Z\otimes Q+Z\otimes Q^{}=Z\otimes\mathbbm{1},\nonumber\\
      &&Z^2X\otimes Q+(Z^2X)^T\otimes Q^{\perp}.
\end{eqnarray}
Depending on the value of $r_{1,k}\neq 0$, both \eqref{Bobs} and \eqref{Charlie} are used in Eqs. \eqref{Copt1} and \eqref{Copt2}.

\item \textbf{($D_w^{(i)}$ observables)} By applying $\mathcal{U}_4=V_1V_3\otimes Q+(V_1V_3)^*\otimes Q^{\perp} $ to the above observables (\ref{Bobs}) one can 
bring them to the following form
\begin{eqnarray}
      &&Z\otimes Q+Z\otimes Q^{\perp}=Z\otimes\mathbbm{1},\nonumber\\
      &&X\otimes Q+X^T\otimes Q^{\perp},
\end{eqnarray}
which is used in Eqs. \eqref{Dopt1} and \eqref{Dopt2}.
\end{itemize}
\end{cor}

\input{graph_biblio.bbl}

\end{document}